\pdfoutput=1
\documentclass[runningheads]{llncs}

\usepackage{soul}
\usepackage{url}
\usepackage[hidelinks]{hyperref}
\usepackage[utf8]{inputenc}
\usepackage[small]{caption}
\usepackage{graphicx}
\usepackage{amsmath}
\usepackage{booktabs}
\usepackage{algorithm}
\usepackage{algorithmic}
\usepackage{fullpage}
\usepackage[numbers]{natbib}
\usepackage{wrapfig}

\usepackage[font={small}]{caption}

\usepackage{xcolor}

\usepackage{natbib}

\usepackage{macros}

\usepackage{subfig}
\begin{document}
\title{Mix and Match:\\ Markov Chains \& Mixing Times for Matching in Rideshare}

\author{Michael Curry\inst{1} \and
John P. Dickerson \inst{1} \and
Karthik Abinav Sankararaman \inst{1, 2} \and
Aravind Srinivasan \inst{1} \and
Yuhao Wan \inst{1, 3} \and
Pan Xu \inst{1, 4}
}
\authorrunning{CDSSWX '19}
\institute{University of Maryland, College Park MD USA 
\and
Facebook, Inc., Menlo Park CA USA 
\and
University of Washington, Seattle WA USA 
\and
New Jersey Institute of Technology, Newark NJ USA 
\\
\email{\{curry,john,srin\}@cs.umd.edu} \\
\email{\{karthikabinavs,yuhao.diane.wan,panxu0\}@gmail.com}
}

\maketitle
\begin{abstract}
	Rideshare platforms such as Uber and Lyft dynamically dispatch drivers to match riders' requests.  We model the dispatching process in rideshare as a Markov chain that takes into account the geographic mobility of both drivers and riders over time. Prior work explores dispatch policies in the limit of such Markov chains; we characterize when this limit assumption is valid, under a variety of natural dispatch policies.  We give explicit bounds on convergence in general, and exact (including constants) convergence rates for special cases.  Then, on simulated and real transit data, we show that our bounds characterize convergence rates---even when the necessary theoretical assumptions are relaxed. Additionally these policies compare well against a standard reinforcement learning algorithm which optimizes for profit without any convergence properties.

\end{abstract}

\section{Introduction}
\label{sec:intro}
Rideshare firms such as Uber, Lyft, and Didi Chuxing dynamically match riders to drivers via an online, digital platform.  Riders request a driver through an online portal or mobile app; a driver is matched by the platform to a rider based on geographic proximity, driver preferences, pricing, and other factors.  The rideshare driver then picks up the rider at her request location, transfers her to her destination, and reenters the platform to be matched again---albeit at a new geographic location.  Part of the larger \emph{sharing economy}, rideshare firms are increasingly competitive against traditional taxi services due to their ease of use, lower pricing, and immediacy of service~\citep{hahn2017ridesharing}.

Matching riders to rideshare drivers is nontrivial.  While the core process is a form of the well-studied online matching problem~\citep{mehta2013online}, current models developed in the EconCS, AI, and Operations Research communities~\citep{ozkan2017dynamic,AAAI18,sid18} do not completely capture the mobile aspects of both the drivers and riders.  Drivers \emph{and} riders are agents who move about a constrained space (e.g., city streets), becoming (in)active periodically due to the matching process.  When a platform receives a request, it must make a near-real-time dispatch decision amongst nearby drivers who are available at the current time.  The platform's goal is to maximize an objective (e.g., revenue or throughput) by servicing requests in an online fashion, subject to various real-world constraints and challenges like setting prices, predicting supply and demand, fairness considerations, competing with other firms, and so on~\citep{cachon2017role,laptev2017time,banerjee2017pricing}.

In this paper, we study the dynamics of the nascent rideshare market under different \emph{dispatch strategies}.  Recent work uses Markov chains to model complex ride-sharing dynamics in a \emph{closed-world} system---that is, a system with a fixed total supply of cars~\citep{ozkan2017dynamic,banerjee2017pricing,sid18}.  These assume the Markov chains reach their stationary distributions quickly, and thus all prior work optimizes for dispatch strategies in the limit.  As a complement, the present paper characterizes---theoretically and empirically---when that limit assumption is valid under a variety of natural strategies.

\xhdr{Our contributions.}
The main contribution in this paper is \emph{to show both theoretically and empirically the convergence rate of many natural policies to the stationary distribution}. First, we model the dispatching problem in rideshare platforms as a Markov chain. The number of states in this Markov chain is exponential in the natural size of the problem; thus, unless the chain is \emph{rapidly mixing}, the time taken to reach the stationary distribution is prohibitively large. Next, we consider two large natural classes of strategies and study the evolution of the driver distribution theoretically. We show that the Markov chains are rapidly mixing and give explicit bounds on the convergence rates. Then, we consider a special case of uniform arrival rates and compute the convergence rates \emph{exactly}, including the constants. Finally, we conduct experiments on both simulated as well as a real-world large-scale dataset to corroborate our findings. In particular, even when the assumptions needed by the theory do not necessarily hold, simulations show that the convergence behavior does not change drastically, and experiments on real data show that the theory gives direct insight on the convergence properties in practice. Additionally, compared against a standard RL algorithm, these policies perform similarly. Hence, our policies are simpler and efficient to run with theoretical convergence guarantees while performing almost as good as more complicated algorithms without such properties.

\section{Preliminaries}
\label{sec:prelims}
In this section, we define the formal model used throughout the paper.  We begin with a brief primer on Markov chains, and then show how Markov chains can be used to model rideshare markets.

\begin{definition}[Markov chain]
  A \emph{Markov chain} $\sM$ is defined by a state space $\Omega$ and a transition matrix $P$.  $P(\x, \y)$ represents the probability of reaching state $\y \in \Omega$, in one step, from state $\x \in \Omega$. $P(\x, \y)$ does not depend on states which the process was in prior to $\x$ (\ie it obeys the \emph{Markov property}).
\end{definition}

A central concept in Markov chain analysis is the notion of a \emph{limiting} distribution and a \emph{stationary} distribution.
For two distributions $\mu$ and $\nu$ on the state space $\Ome$, the \emph{total variation distance} between $\mu$ and $\nu$ is defined as $||\mu-\nu||_{TV}=\frac{1}{2} \sum_{\x \in \Ome} |\mu(\x)-\nu(\x)|$. A distribution $\pi^*$ is said to be a \emph{stationary distribution} if $\pi^*=\pi^* P$. A distribution $\pi^*$ is said to be a \emph{limiting distribution} if for every initial state $\x$, we have that $\lim_{t \rightarrow \infty}||P^{t}(\x, \cdot)-\pi^*||_{TV}=0$, where $P^t(\x, \cdot)$ denotes that distribution of states after $t$ steps, starting at $\x$. 	

To use Markov chains as an algorithmic tool, we need to understand the \emph{rate} of convergence to the stationary distribution, commonly called its \emph{mixing time}. 

\xhdr{Mixing Time $\tau(\ep)$.} Consider an irreducible and aperiodic Markov Chain $\sM$ with stationary distribution $\pi^*$. For a given $t$, let $d(t)=\max_{\x \in \Ome} ||P^t(\x, \cdot)-\pi^*||_{TV}$. The mixing time of $\cM$ is defined as $\tau(\ep)=\min\{t: d(t') \le \ep,\forall t' \ge t\}$. We say $\sM$ is \emph{rapidly mixing} if $\tau(\ep)=O\Big(poly\big(\log \frac{|\Ome|}{\ep} \big)\Big)$.

In this paper, we consider a class of strategies (as motivated by, \eg~\cite{AAAI18}) for the rideshare problem and cast it as a natural Markov chain. We then consider an objective function which depends on the limiting distribution of this chain and study convergence properties of that objective function, using the mixing properties of the Markov chain. We use Markov chains and relevant tools as the central concepts in this paper. \citet{MCBook} details classical results in this space. The first \emph{algorithmic} usage of Markov chain Monte Carlo (MCMC) methods can be traced back to the classical works of \citet{hastings1970monte,metropolis1949monte,metropolis1953equation}, commonly known as the Metropolis-Hastings algorithm. MCMC methods are a powerful tool in machine learning and we refer the reader to the survey~\cite{andrieu2003introduction}.

\xhdr{A Markov chain model of rideshare.} We now define our Markov-chain-based model of rideshare.  Consider a two-dimensional grid $\cU$ consisting of $n$ points (e.g., geographic locations). A \emph{request type} $r=(u,u')$, represented by an \emph{ordered}  pair of points, is a set of requests that start and end at locations $u \in \cU$ and  $u' \in \cU$, respectively. Let $\cR =\{r=(u,u')| u \in \cU, u' \in \cU\}$ be the set of all request types. Note that we allow request types $r=(u,u)$---that is, a request that both starts and ends at the same point $u \in \cU$. This is just for notational convenience.
		
Given a time horizon $T$, at each time (or round) $t \in [T] \doteq\{1,2,\ldots,T\}$, a request of type $r$ is sampled from $\cR$ with probability $p_r$.\footnote{We have $\sum_{r \in \cR} p_r \le 1$. Thus, with probability $1-\sum_{r \in \cR} p_r$, there is no request in any given time.} Our goal is to design a \emph{matching} (or dispatching) scheme---assigning a driver (or car) to a request---that maximizes an overall objective after $T$ rounds. In this paper, we assume that the sampling distribution $\{p_r\}$ in every round is \emph{identical and independently distributed} (IID) but \emph{unknown} to the algorithms. For notational simplicity, we also use $r$ to denote a specific online request of type $r$ when the context is clear.
	
\xhdr{Dispatching policy.} Suppose the system has $m$ identical drivers. We characterize the state of the system by a vector $\x \in \mathbb{Z}_{+}^n$,  where $x_u$ denotes the number of drivers in location $u$.  Then, we can construct a Markov chain with state space $\Omega=\{\x: x_u \in \{0,1,\ldots,c\}, \sum_u x_u=m\}$.  In this paper, we assume that $m \ll c \cdot n$ for some constant capacity $c$. A \emph{dispatching policy} (or strategy) $\sigma$ is a mapping from $\Omega \times \cR$ to $\cU \cup \{\emptyset\}$ such that at time $t$ and a state $\ind{\X}{t}{}$, when a request $r$ comes, $\sigma$ assigns that request $r$ to a potential driver at location $u_{\sigma}=\sigma(\ind{\X}{t}{}, r)$.\footnote{The choice of $u_\sigma$ can be random since $\sigma$ can be a randomized policy.} Here $u_{\sigma}=\emptyset$ denotes that the policy $\sigma$ rejects request $r$. We say $\sigma$ successfully addresses the request $r=(u,v)$ at $t$ if $\ind{X}{t}{u_{\sigma}} \ge 1$ and $\ind{X}{t}{v}<c$.  If $\sigma$ successfully addresses a request $r$, it receives a profit $w_r$.

	\xhdr{Neighborhood of $u$.} For each point $u \in \cU$, let $\cN(u)$ be the set of neighbors of $u$ with Manhattan distance\footnote{It is not critical for our purposes, but the experiments use New York city and road-distance is measured in Manhattan distance.} exactly $1$ to $u$, \ie $\cN(u)=\{u' \in \cU:  |u-u'|_M = 1\}$, where $|u-u'|_M$ denotes the Manhattan distance between $u$ and $u'$. We can assign each request with origin $u$ only to a driver in the set $\{u\} \cup \mathcal{N}(u)$.

	\xhdr{Objective functions.} For a given policy $\sigma$, let $W(\sigma,t)$ be the expected profit obtained by the policy $\sigma$ at time $t$. Denote $\I_{\X} := \I ( \ind{X}{t}{u_{\sigma}} \ge 1, \ind{X}{t}{v}<c)$ which is the indicator for the event $ \ind{X}{t}{u_{\sigma}} \ge 1$ and $ \ind{X}{t}{v}<c$. Thus the expected performance of $\sigma$ at time $t$ and the expected average performance of $\sigma$ over $T$ rounds is, 
 \begin{eqnarray}
		W(\sigma,t) = \textstyle \E\Big[\sum_{r=(u,v)\in \cR} p_r \cdot w_r \cdot  \I_{\X} \Big], &\label{eqn:OBJ-t}\\ 
 	\OBJ(\sigma, T)= \textstyle \frac{1}{T}\sum_{t=1}^{T}W(\sigma,t).& \label{eqn:OBJ}
 \end{eqnarray}
 respectively. The randomness in the state $\ind{\X}{t}{}$ depends on two sources: the random arrival of requests from distribution $\{p_r\}$ and any internal randomness used in the execution of a randomized policy $\sigma$. 
 
 In this paper, our objective is to study a class of dispatching strategies that are both \emph{effective} and \emph{stable}. Specifically, suppose $\x_0$ is the initial state, $\bp=\{p_r| r \in \cR\}$ is the arrival distribution of the request types in each round, and $\bw=\{w_r | r\in \cR\}$ is the profit vector for the request types. We then wish to answer the following questions.
 
 	\begin{enumerate}
 		\item \textbf{Effectiveness.} For a given instance $\cI=(\x_0, \bp, \bw, T)$, which non-adaptive strategy $\sigma$ maximizes $W(\sigma,t)$ and $\OBJ(\sigma, T)$?
 		\item \textbf{Stability.} For a given non-adaptive strategy $\sigma$, do the limits $\lim_{t \rightarrow \infty}W(\sigma,t)$ and $\lim_{T \rightarrow \infty} \OBJ(\sigma, T)$ exist? If so, how fast do they converge to the respective limiting values? 
  	\end{enumerate}

\section{Assumptions \& Related Work}
\label{sec:assumptions-and-rw}

Rideshare is a recent and popular innovation; thus, the body of literature surrounding this paradigm is young and quickly growing.  With that in mind, we now explicitly motivate and list the assumptions we make in our paper, and then place our model in the greater body of related work.

\xhdr{Model assumptions.} First, we assume that the number of drivers in the system remains constant in the $T$ rounds with no new driver either joining or leaving the system. This assumption is justified because (1) $T$ rounds in the online phase is typically restricted to a few hours and (2) almost all trips are local (as described later in our experimental section on real data).

\begin{figure}[!h]
\centering
  \includegraphics[scale=0.4]{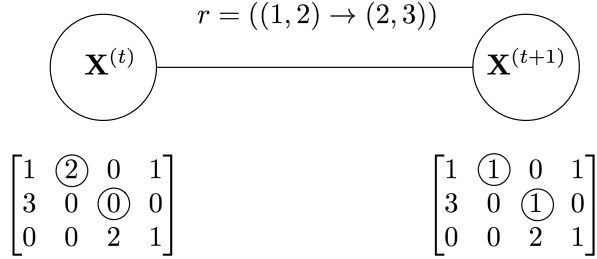}
  \setlength{\belowcaptionskip}{-15pt}
  \caption{A single state transition in the Markov chain; driver matched at $(1, 2)$, trip ends at $(2, 3)$.}
  \label{fig:example}
\end{figure}	

Second, at time $t$ when a request $r=(u,u')$ comes and is assigned to a driver at location $v$, we assume that (1) the system obtains a profit $w_{r}>0$,  which is proportional to the distance between $u$ and $u'$; (2) at time $t+1$, the number of drivers at $v$ reduces by $1$ and the number at $u'$ increases by $1$ (\ie the request is completed instantaneously). From the real dataset, we observe that most trips in the Manhattan area are local and are completed within $15$ mins (\ie a short time period). Hence things do not change drastically by making this simplifying assumption. 

Third, we assume that all locations have the same capacity $c \in \mathbb{Z}_{+}$, an upper bound on the total number of drivers that can be present in each location at any time. Here $c$ captures the maximum number of drivers allowed at any single location. 

Fourth, we make the following global \emph{hot-spot assumption} about the requests. There exists a location $u^* \in \cU$ such that $p_{(u^*,u)}>0$ and $p_{(u,u^*)}>0$ for all $u \in \cU, u \neq u^*$. We call this location $u^*$ the ``hot-spot''. This assumption naturally holds in many real scenarios (e.g., cell $(10, 5)$ in Figure~\ref{fig:heatmaps}), since most big cities have busy central locations (\eg the Empire State Building in New York City (NYC)), or have flow to and from an airport (\eg JFK or LGA, in the case of NYC). 

\begin{figure}%
                \centering
     \subfloat{{\includegraphics[scale=0.5]{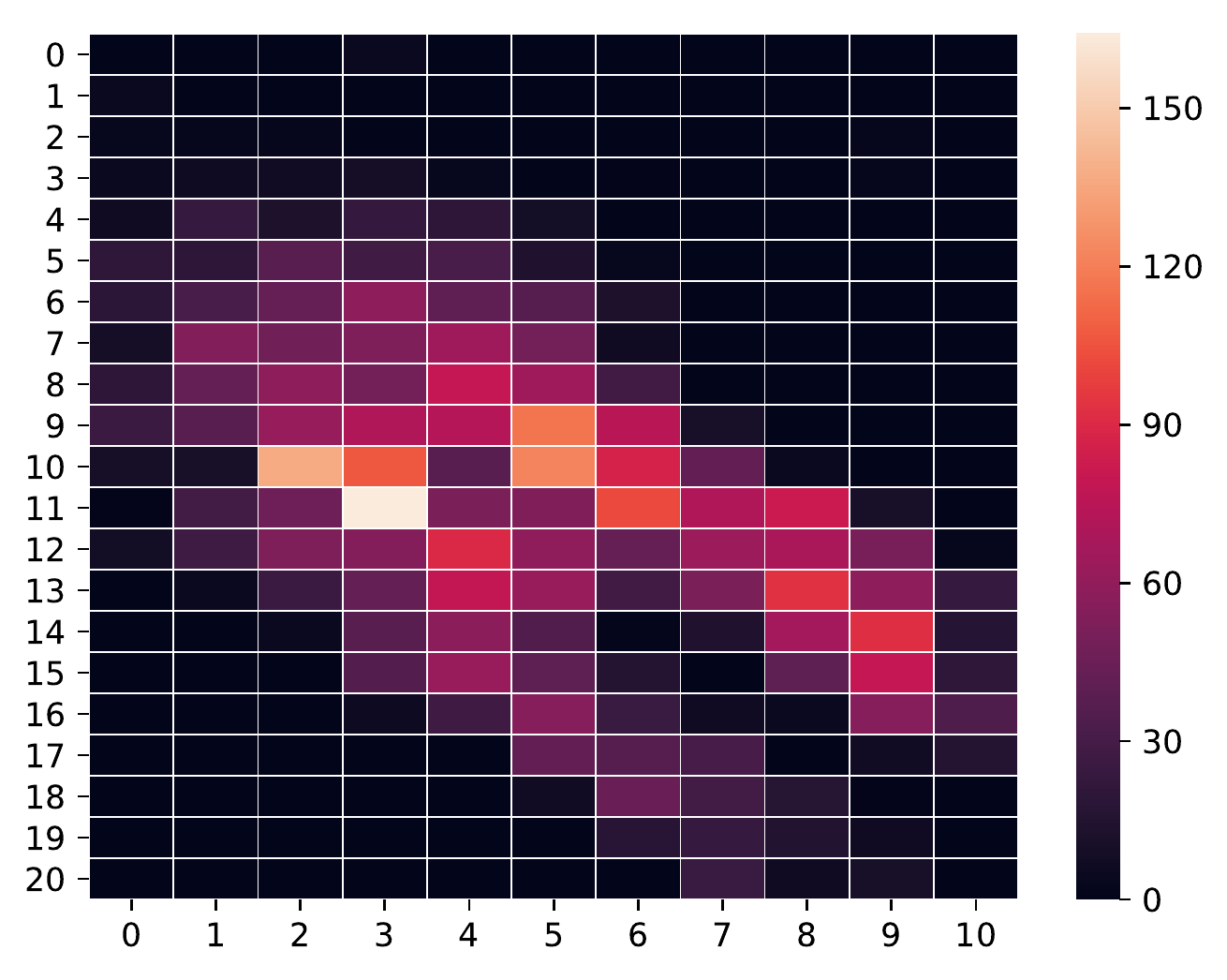} }}%
     \subfloat{ {\includegraphics[scale=0.5]{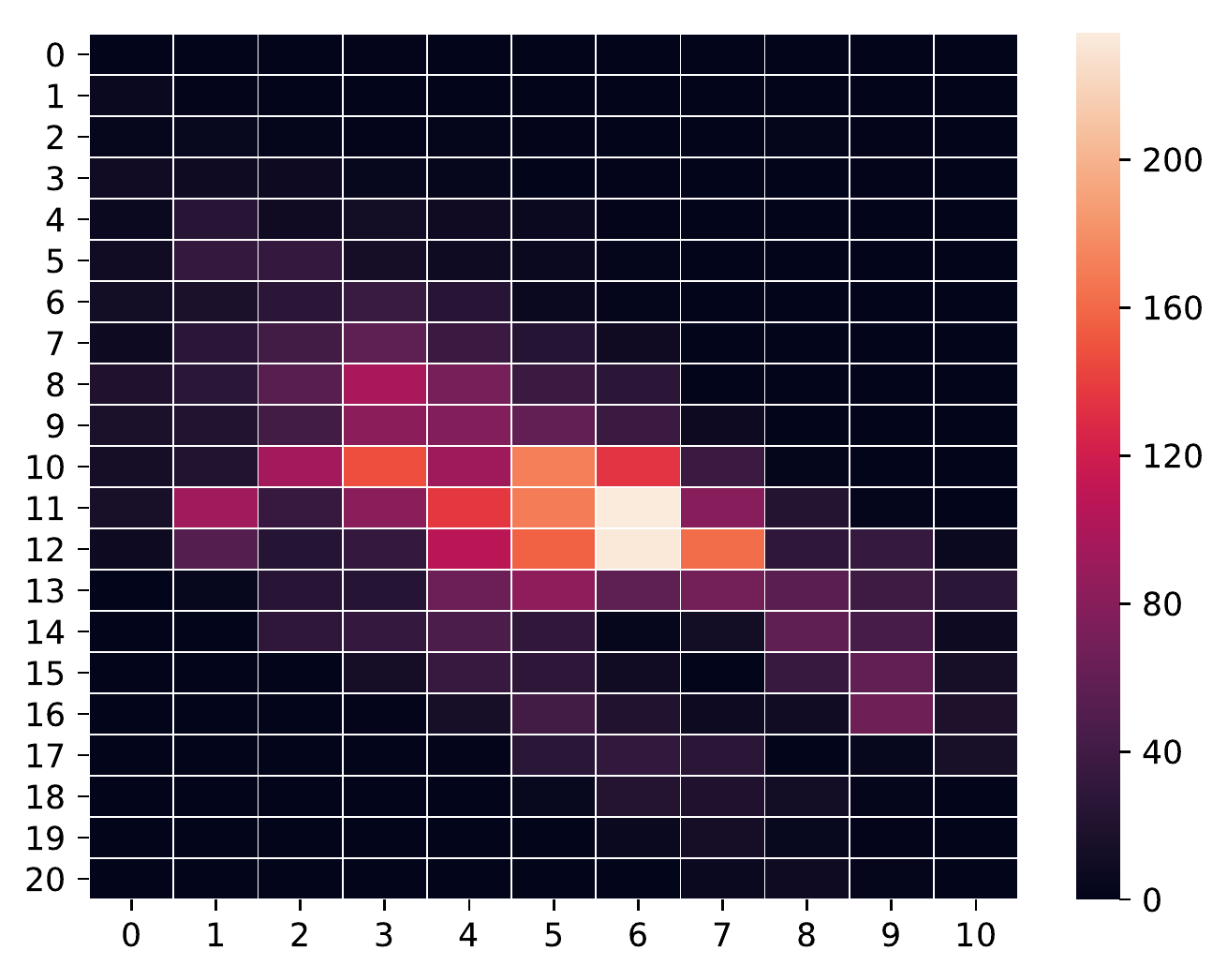} }}%
    \caption{Heatmap of starting (left) and ending (right) location of rides in NYC (7:00AM--11:00AM). X-axis, Y-axis: numbered from 0-19. Each cell represents a location. Heat index 0-black to 160-white.}%
    \label{fig:heatmaps}%
\end{figure}

Finally, we assume IID arrivals of online requests.  This is necessary for convergence---even for very simple dispatching policies. Consider the following example with two locations $u$ and $v$ in the system and a single driver. Define request types $r_a=(u,v)$ and $r_b=(v,u)$. Suppose the two online requests $r_a$ and $r_b$ arrive alternatively during odd and even rounds, and our dispatching policy is the simple greedy one: match a request $r$ if there is one driver at the starting location of $r$, otherwise, reject it. Then the Markov chain will admit no stationary distribution, since it has a period of $2$.

It is worth noting that the first, second and fifth assumptions are used in related work~\citep{banerjee2017pricing,sid18}. The third assumption is a generalization of prior works which all considered $c=m$.

\xhdr{Related work.} Research in rideshare platforms and similar allocation problems is an active area of research within multiple fields, including computer science, operations research and transportation engineering. State-independent policies were studied previously using theory from control and queuing systems~\citep{ozkan2017dynamic,braverman2016empty,banerjee2017pricing}. Apart from using Markov-chain methods, many allocation and scheduling problems have been studied in the rideshare context using methods from combinatorial optimization and machine learning (\eg \cite{AAAI18,WWW18,jia2017optimization}). A large-scale mathematical and empirical study on the number of cars and the optimality of waiting times was recently analyzed by \citet{alonso2017demand}. In our work, we do not consider the waiting times; indeed, we assume that the match and the trips happen instantaneously. The role of pricing in the dynamics of drivers in rideshare platforms is also an active area of research in computational economics and AI/ML (\eg \cite{ma2018spatio,dpAAAI18,banerjee2016dynamic}). Although rideshare platforms face challenges that are unique and different from \emph{bike-sharing} ecosystems, there are some similarities between the two. Both of these deal with a matching market where the agents are constantly moving around and hence it is important to characterize the \emph{flow} patterns of these agents. A number of works study the prediction aspects as well as the dynamics of flow patterns in both bike-sharing (\eg \cite{o2015data,ghosh2017dynamic}) as well as in ridesharing (\eg \cite{wen2017rebalancing,verma2017augmenting,laptev2017time,yao2018deep,qu2014cost}) platforms.
	
Our problem is a form of \emph{online matching in dynamic environments}, which is an active area of research within the AI/ML community. In particular, \cite{dickerson2015futurematch,lowalekar2016online,tong2016online,ICDE17} have studied algorithms for matching in various dynamic markets such as kidney exchange, spatial crowdsourcing, labor markets, and so on. Online matching in \emph{static} environments has been extensively studied in the literature; see \citet{mehta2013online} for a detailed survey. We use a reinforcement learning algorithm as a baseline in our experiments. Apart from a recent work by \citet{RLrideshare}, which explores the fleet management problem in ride-share, this tool has not been explored much in this domain. However, reinforcement learning (RL) has found success in other applications and we refer the reader to a recent survey~\cite{li2017deep}.

\xhdr{Comparison with related concurrent work.} In a recent work, \citet{sid18} proposed an approximation framework for pricing in rideshare platforms. Our model shares many characteristics with theirs (\ie the first, second and fifth assumptions).  The main difference is the focus of the paper. In particular, they are interested in finding the optimal (online) dispatching policy that has a \emph{good} approximation ratio with respect to expected revenue at the stationary distribution. Thus they consider the \emph{known IID} arrival assumption for online requests while we consider the stronger \emph{unknown} IID arrival assumption. Moreover they assume that the number of cars tends to infinity, and seek approximation ratios in this limit. The main focus of this paper is instead to characterize the \emph{rate} of convergence to the stationary distribution. We seek to understand when using the expected reward in the stationary distribution is a good measure, especially if the rate of convergence is slow and/or non-existent.

\xhdr{From dispatching to pricing.} The techniques in this paper can be extended to the scenario of dynamic pricing. Suppose that the system cannot directly match riders to drivers, but can set prices for every ride. Then, we can modify the schemes in this paper under certain conditions. Consider the pricing model as studied in \citet{banerjee2017pricing}. Every rider arrives with a value, which can be viewed as a sample from a distribution associated with its request type.  The system then sets a price (this can be arbitrarily correlated with the valuation distribution and may depend on both the time and request type but not the current state of the system) and the rider accepts the service if and only if their value is greater than the offered price. This process can be incorporated into the arrival rates of the requests and hence one can obtain similar convergence rates for this dynamic pricing problem. For brevity, we consider the problem of dispatching as the running example throughout this paper. 

\section{Dispatching Polices}

In this section, we present our dispatching policies for matching riders to rideshare drivers.  Specifically, we present two policies namely $\nadap(\alp)$ and $\rand(\phi)$. The main ideas of these are as follows. 

When a request $r=(u,u')$ arrives, $\nadap(\alp)$ checks the availability of drivers in locations $u$ and each of its neighbors in $\cN(u)$ with respective probabilities $\alp$ and $\frac{1-\alp}{4}$, where $0<\alp<1$ is a parameter. (Technically, nodes $u$ on the boundary of the grid have less than $4$ neighbors; we just assume that we do nothing if the realization is an invalid neighbor.) 

When a request $r=(u,u')$ arrives, $\rand(\phi)$ will first check the availability of drivers at $u$ and if there aren't any, then checks the neighbors in $\cN(u)$ following a given order $\phi$. Below, we present the theoretical results on their convergence rates.

\begin{theorem}[Main Convergence Theorem]
	\label{thm:mainTheorem}
For any policy $\nadap(\alp)$ with $\alp>0$ and $\rand(\phi)$ starting with any given initial state, the objective functions defined in \eqref{eqn:OBJ-t} and \eqref{eqn:OBJ} both converge to the same value with rates $\Theta(\beta^t)$ and $\Theta(T^{-1})$, respectively, for some $\beta \in (0,1)$ independent of $T$ (possibly dependent on the other input parameters such as $m$ and $n$).
\end{theorem}
Theorem~\ref{thm:mainTheorem} refers to the absolute error bound of the objectives in \eqref{eqn:OBJ-t} and \eqref{eqn:OBJ} from its limit value. In particular, Theorem \ref{thm:mainTheorem} states that as the number of rounds increases, the expected revenue in every timestep converges to a stable value \emph{exponentially} quickly. 

	\xhdr{Markov chains $\sM(\nadap(\alp))$ and $\sM(\rand(\phi))$.}  We make the simplifying assumption that for every $u \in \mathcal{U}$, the number of neighbors $|\cN(u)|=4$. Recall that $\Ome=\{\x: x_u \in \{0,1,\ldots, c\}, \sum_u x_u=m\}$. Define a Markov chain $\sM(\nadap(\alp))$ for $\nadap(\alp)$ over $\Ome$ as follows. For any \emph{ordered} pair $(\x,\y)$ where $\x, \y \in \Omega$, we say $(\x,\y)$ is a $(u,u')$-neighbor if $y_u=x_u-1$, $y_{u'}=x_{u'}+1$, and $x_{v}=y_v$ for every $v \notin \{u,u'\}$. The policy $\nadap(\alp)$  induces the following transition probability between all possible $(u,u')$-neighbors. 
		    \[
    			 q(u,u') \doteq \alp p_{(u,u')}+ \frac{1-\alp}{4} \sum_{v \in \cN(u)} p_{(v,u')}.
    		\]
  			Thus the transition matrix for the Markov chain $\sM(\nadap(\alp))$ can be defined as follows.
			\begin{equation}
				\label{eqn:trans-mat-1}
				P_\alp(\x,\y)= q(u,u') \quad \text{iff $(\x,\y)$ is a $(u,u')$ neighbor}
 			\end{equation}

Similarly, we can define a Markov chain $\sM(\rand(\phi))$ over $\Ome$ for $\rand(\phi)$ as follows.
For a given random permutation $\phi$  and $v \in \cN(u)$, let $\phi(v)$ be the random order assigned by $\phi$. Assume that $\rand(\phi)$ checks the neighbors $\cN(u)$ in the order $(\phi^{-1}(1), \phi^{-1}(2),\phi^{-1}(3), \phi^{-1}(4))$. For a given state $\x$, location $u$ and random order $\phi$, the respective \emph{supportive neighbor} of $u$ is defined as $\cN_{\x,\phi}(u)=\{v \in \cN(u): x_v=0, x_{v'}=0, \forall v' \in \cN(v), \phi(v')<\phi(u)\}$.
In other words, $\cN_{\x, \phi}(u)$ includes those neighbors $v$ of $u$ such that when we are at state $\x$ while a request $r$ with origin $v$ arrives, $\rand(\phi)$ will surely match $r$ to a driver at $u$ if present. The resultant transition matrix for the Markov chain $\sM(\rand(\phi))$ is non-zero iff $(\x,\y)$ is a $(u,u')$ neighbor. The non-zero value is as follows.
	\begin{equation}
		\label{eqn:trans-mat-2}
		P_\phi(\x,\y)= p_{(u,u')}+\sum_{{r=(v,u'): v \in \cN_{\x,\phi}(u)}} p_r.
 	\end{equation}

We use the following lemma about $\sM(\nadap(\alp))$ and $\sM(\rand(\phi))$. The proof is deferred to Appendix~\ref{sup:main}.
		
		\begin{lemma}
			\label{lem:irre}
			Under the hot-spot assumption, Markov chains $\sM(\nadap(\alp))$  and $\sM(\rand(\phi))$ defined in \eqref{eqn:trans-mat-1} and \eqref{eqn:trans-mat-2} are both irreducible and aperiodic for any given $\alp>0$ and permutation function $\phi$. Thus, from Theorem~\ref{thm:conv} in the primer on Markov chains, both $\sM(\nadap(\alp))$  and $\sM(\rand(\phi))$ admit a unique limiting distribution $\pi^*$, which coincides with the unique stationary distribution. 
		\end{lemma}
		
\subsection{Proof of Theorem \ref{thm:mainTheorem}} \label{sec:proof-1}
	\label{sup:thm-sim}

We prove this theorem by first showing the convergence rate for the objective~\eqref{eqn:OBJ-t} for both the policies. Finally we show how this can be adapted to give the convergence rates for the objective~\eqref{eqn:OBJ}.

\xhdr{Convergence rate of $\nadap(\alp)$ for Objective~\eqref{eqn:OBJ-t}.} We first show the convergence rates for $\sM(\nadap(\alp))$. From Lemma \ref{lem:irre}, we have that $\sM(\nadap(\alp))$ is irreducible and aperiodic when $\alp>0$ and thus, it admits a unique limiting distribution (say $\pi^*$), regardless of the initial state. Suppose $\X^{(t)}$ is the random state after $t$ steps (let the starting state be $\x_0$), which follows the distribution $P^t(\vec{x}^{(0)}, .)$. Let $\X^{(\infty)}$ be the random state when $T \rightarrow \infty$, which follows the distribution $\pi^*$. For each $u,v \in \cU$, define $\gam_{u,v}$ as the probability that there is at least one driver at location $u$ and there are less than $c$ drivers at $v$ in $\pi^*$. Thus by definition we have,
	\[
		\gam_{u,v}=\Pr[\ind{X}{\infty}{u} \ge 1, \ind{X}{\infty}{v} < c]=\sum_{\x \in \Ome: x_u \ge 1, x_v<c} \pi^*(\x).
	\]

	For each $u,v$ and $t$, define $\gamma_{u,v}^{(t)} :=\Pr[\ind{X}{t}{u}\ge 1, \ind{X}{t}{v}<c]$. We prove the following Lemma~\ref{lem:thm-sim} which we later use to prove the convergence rate in the main theorem.

\begin{lemma}
	\label{lem:thm-sim}
There exists a scalar $C > 0$, independent of $t$ and $\beta \in (0,1)$ such that $|\gam_{u,v}^{(t)}-\gam_{u,v}| \le 2C \beta^t$ for any $u,v \in \cU$ and initial state $\x_0$.
\end{lemma}
\vspace{-0.1in}
\begin{proof}
Let $P^t(\x_0,\y)=\Pr[\X^{(t)}=\y]$ be the probability of reaching state $\y$ after $t$ steps with initial state $\x_0$. Thus we have,
\begin{align}
|\gam_{u,v}^{(t)}-\gam_{u,v}|  &= \Big| \sum_{\y: y_u \ge 1, y_v<c} \Pr[\X^{(t)}=\y] -\gam_{u,v}\Big|\\
&= \Big|   \sum_{\y: y_u \ge 1, y_v<c} P^t(\x_0,\y)-  \sum_{\vec{z} \in \Ome: z_u \ge 1,z_v<c} \pi^*(\vec{z}) \Big| \\
& \le  \sum_{\y: y_u \ge 1, y_v<c}  \Big| P^t(\x_0,\y)-\pi^*(\y)  \Big| \\
& \le 2| P^t(\x_0,\cdot)-\pi^* |_{TV} \le 2 C \beta^t.  \label{ineq:lem-gam}
\end{align}

In equation \eqref{ineq:lem-gam}, the first inequality  is from the definition of total variation between two distributions while the second inequality is from the convergence theorem 4.9 in \cite{MCBook} (also Theorem~\ref{thm:conv} in primer on markov chains in Appendix~\ref{sec:sup-basic}).
\end{proof}

\xhdr{Convergence rate of $\sM(\rand(\phi))$ for 	Objective~\eqref{eqn:OBJ-t}.} Recall that we proved in Lemma~\ref{lem:irre} that $\sM(\rand(\phi))$ admits a unique limiting distribution. For notational convenience, we overload $\pi^*$ and $P$ to denote the limiting distribution and transition matrix of $\sM(\rand(\phi))$. Set $\overline{\cN}(u)= \{u\}\cup \cN(u)$. Let $\Y^{(t)}$ be the random state in $\sM(\rand(\phi))$ after $t$ steps starting with $\y_0$. For a given $u$ and $v$ let $\eta_{u,v}$ and $\eta_{u,v}^{(t)}$ be the respective probabilities, in the limiting distribution $\pi^*$ and $\Y^{(t)}$, that there is at least one driver in $\overline{\cN}(u)$ and less than $c$ drivers at $v$. Therefore by definition we have that,
\begin{align*}
\eta_{u,v}&=\Pr\Big[ \bigvee_{v \in \overline{\cN}(u)} \Big( Y^{(\infty)}_v \ge 1\Big), ~Y^{(\infty)}_v<c \Big]=\sum_{\x: \sum_{v \in \overline{\cN}(u)} x_v \ge 1 , x_v<c}\pi^*(\x),\\
\eta_{u,v}^{(t)} &=\Pr\Big[ \bigvee_{v \in \overline{\cN}(u)} \Big( Y^{(t)}_v \ge 1\Big), ~Y^{(t)}_v<c\Big]=\sum_{\x: \sum_{v \in \overline{\cN}(u)} x_v \ge 1, x_v<c }P^t(\y_0,\x).
\end{align*}

Similar to Lemma \ref{lem:thm-sim}, we have the following lemma for $\sM(\rand(\phi))$.
\begin{lemma}\label{lem:thm-rand}
There exists a scalar $C > 0$ independent of $t$ and $\beta \in (0,1)$ such that $|\eta_{u,v}^{(t)}-\eta_{u,v}| \le 2C \beta^t$ for any $u,v \in \cU$ and initial state $\y_0$.
\end{lemma}

The proof of Lemma \ref{lem:thm-rand} is essentially the same as that of Lemma~\ref{lem:thm-sim} since it does not use the properties of $\nadap(\alp)$.

Now we have all ingredients to prove main Theorem \ref{thm:mainTheorem}.
\begin{proof}

We focus on applying Lemma  \ref{lem:thm-sim} to prove the results for $\sM(\nadap(\alp))$ in Theorem~\ref{thm:mainTheorem}. A similar analysis applies to the case $\sM(\rand(\phi))$ by using  Lemma \ref{lem:thm-rand}.

Consider the case $\sigma=\nadap(\alp)$ and a given $r=(u,u')$. Policy $\nadap(\alp)$ implies that $u_{\sigma}=\sigma(\X^{(t)}, r)$ will be equal to $u$ with probability $\alp$ and any neighbor in $\cN(u)$ with probability $\frac{1-\alp}{4}$. Thus, we claim that 
\begin{align*}
\lim_{t \rightarrow \infty}\E\Big[ \I(X^{(t)}_{u_{\sigma}}\ge 1, X^{(t)}_v<c)\Big] &=\alp \Pr[X^{(\infty)}_u \ge 1, X^{(\infty)}_v <c]+\frac{1-\alp}{4}\sum_{k \in \cN(u)}\Pr[X^{(\infty)}_{k} \ge 1, X^{(\infty)}_v <c] \\
&=\alp \gam_{u,v}+ \frac{1-\alp}{4} \sum_{k \in \cN(u)} \gam_{k,v}
\end{align*}
Let $W(\alp)= \lim_{t \rightarrow \infty} W(\nadap(\alp), t) $. Thus we claim that
\begin{align}
 W(\alp) &=\lim_{t \rightarrow \infty}\E\Big[\sum_{r=(u,v) \in \cR} p_r \cdot w_r \cdot  \I \big( \ind{X}{t}{u_{\sigma}} \ge 1, X^{(t)}_v<c\big)\Big] \nonumber \\
 &=\sum_{r=(u,v) \in \cR} p_r \cdot w_r \cdot \Big(  \alp \gam_{u,v}+ \frac{1-\alp}{4} \sum_{k \in \cN(u)} \gam_{k,v} \Big). \label{eq:thm1Later}
\end{align}

Now we bound the convergence rate of $W(\nadap(\alp),t)$. Let $\Delta(t) \doteq | W(\nadap(\alp), t) - W(\alp) |$ (we omit the subscript of $\alp$ in $\Delta(t)$ here) and $\wmax \doteq \max_{r \in \cR} w_r$. We have

\begin{align}
\Delta(t)&=
\Big|\sum_{r=(u,v) \in \cR} p_r \cdot w_r \cdot 
\Big( \alp \Pr[X^{(t)}_u \ge 1, X^{(t)}_v <c]+\frac{1-\alp}{4}\sum_{k \in \cN(u)}\Pr[X^{(t)}_{k} \ge 1,  X^{(t)}_v <c] \Big) \\
& - \sum_{r=(u,v) \in \cR} p_r \cdot w_r \cdot \Big(  \alp \gam_{u,v}+ \frac{1-\alp}{4} \sum_{k \in \cN(u)} \gam_{k,v} \Big) \Big|\\
&\le  \sum_{r=(u,v) \in \cR} p_r \cdot w_r \cdot 
\left( \alp \big| \gam_{u,v}^{(t)} -\gam_{u,v} \big|+ \frac{1-\alp}{4} \sum_{k \in \cN(u)}  \big| \gam_{k,v}^{(t)}-\gam_{k,v}\big|
\right)\\
& \le\sum_{r=(u,v) \in \cR} p_r \cdot w_r \cdot (2C \beta^t)\label{ineq:proof-thm1-a}\\
& \le  2 \wmax C \beta^t. \label{ineq:proof-thm1-b}
\end{align}
Inequality \eqref{ineq:proof-thm1-a} is a direct application of Lemma~\ref{lem:thm-sim} and the fact that $|\cN(u)| \le 4$ for each $u$; Inequality~\eqref{ineq:proof-thm1-b} is due to the fact that $\sum_{r \in \cR} p_r\le 1$.

Recall that $\OBJ(\alp)=\lim_{T\rightarrow \infty} \OBJ(\nadap(\alp), T)$.
We can verify that $\OBJ(\alp)=W(\alp)$. Now we bound the convergence rate of $\OBJ(\nadap(\alp),T)$.  Let $\widehat{\Delta}(T) \doteq | \OBJ(\nadap(\alp), T) - \OBJ(\alp) |$. Notice that we have the following.
\begin{align}
\label{eq:ProofEq2}
\widehat{\Delta}(T)& \le \frac{1}{T}\sum_{t=1}^{T} \Delta(t) =   \frac{1}{T}\sum_{t=1}^{T} 2\wmax C \beta^t \le \frac{2\wmax C}{T} \frac{\beta}{1-\beta}.
\end{align}
This completes the proof.
\end{proof}

\subsection{A Special Case: Uniform Arrivals of Online Requests}
 We now focus on the policy, $\nadap(\alp)$, and
consider a special case when the requests arrive as a uniform sample in each round (\ie $p_r=p$ for all $r \in \cR$). Then we can compute explicit values of $C$ and $\beta$. Note that in this case, the transition matrix for $\sM(\nadap(\alp))$ is $P_\alp (\x,\y)=\alp p_{(u,u')}+ \frac{1-\alp}{4} \sum_{v \in \cN(u)} p_{(v,u')}=p$.

Consider a given instance of the rideshare setting $\cI=(\x_0, \bp, \bw, T)$ and $\nadap(\alp)$. \\
Let $W(\alp) \doteq  \lim_{t \rightarrow \infty} W(\nadap(\alp),t)$,
 and $\OBJ(\alp) \doteq \lim_{T \rightarrow \infty} \OBJ (\nadap(\alp), T)$, where $W(\nadap(\alp),t)$ and $ \OBJ (\nadap(\alp), T)$ are defined in Equations \eqref{eqn:OBJ-t} and \eqref{eqn:OBJ} respectively. 
From Lemma~\ref{lem:irre}, we have that the limits in Equations~\eqref{eqn:OBJ-t} and \eqref{eqn:OBJ} both exist and are the same for any given $\alp>0$ with the limiting value as in Equation~\eqref{eqn:obj-sim}.  Define $\Delta_\alp(t) \doteq |W(\alp)-W(\nadap(\alp),t)|$ and $\widehat{\Delta}_\alp(T) \doteq |\OBJ(\alp) - \OBJ (\nadap(\alp), T)|$. From Equations~\eqref{eq:thm1Later}, \eqref{ineq:proof-thm1-b} and \eqref{eq:ProofEq2} in the proof of Theorem~\ref{thm:mainTheorem} we have
\begin{align}
\label{eqn:obj-sim}
	& W(\alp) =\OBJ(\alp)  = \sum_{r=(u,v) \in \cR} p_r \cdot w_r \cdot \Big(  \alp \gam_{u,v}+ \frac{1-\alp}{4} \sum_{k \in \cN(u)} \gam_{k,v} \Big). \\
  & \Delta_\alp(t) \le 2 \wmax C \beta^t, ~ \widehat{\Delta}_\alp(T) \le   \frac{1}{T}\sum_{t=1}^{T} \Delta_\alp (t) \le \frac{2\wmax C}{T} \frac{1}{1-\beta}.   \label{eqn:Del-sim}  
\end{align}

 In Equation~\eqref{eqn:obj-sim}, $\gam_{u,v}$ is the probability that location $u$ has at least one driver and location $v$ has fewer than $c$ drivers in the limiting distribution $\pi^*$.  In Inequalities~\eqref{eqn:Del-sim}, $\wmax \doteq \max_{r \in \cR} w_r$ and $C>0$ and $\beta \in (0,1)$ are two scalars which are independent of $t$ and $T$, but may be related to parameters $m$ and $n$.  This yields the result for $\nadap(\alp)$; we can get similar results for the policy $\rand(\phi)$.

Observe that if $(\x,\y)$ is a $(u,u')$-neighbor, then $(\y,\x)$ is a $(u',u)$-neighbor and therefore we have $P_\alp(\y,\x)=p$, implying that $P_\alp$ is symmetric. This implies that $\sM(\nadap(\alp))$ admits a unique uniform limiting distribution $\pi^*$ over $\Ome$ (details in Appendix~\ref{sec:sup-basic}). This enables us to get a closed-form expression for $W(\alp)$ (Theorem~\ref{thm:obj}), and derive explicit values of $C$ and $\beta$ (Theorem~\ref{thm:error}).

\begin{theorem}[Closed-form expression for $W(\alp)$]
\label{thm:obj}
Consider the uniform arrival distribution of the requests such that $p_r=p$ for all $r \in \cR$ and $c=m$. Objectives defined in Equations \eqref{eqn:OBJ-t} and \eqref{eqn:OBJ} for $\nadap(\alp)$ both converge to $\frac{m p}{n+m-1} \sum_{r\in \cR} w_r$ for any given $\alp>0$. 
\end{theorem}

The condition $c=m$ in Theorem \ref{thm:obj} implies that there is no constraint on the maximum number of drivers in any location. The main idea of proof for Theorem~\ref{thm:obj} is to show that for any $(u,v)$, $\gam_{u,v}=\frac{m}{n+m-1}$ for the special case. We defer the full proof of Theorem \ref{thm:obj} to Appendix~\ref{sup:main}.

\begin{theorem}[Explicit values of $C$ and $\beta$]
\label{thm:error}
Consider the uniform arrival distribution of the requests such that $p_r=n^{-2}$ for all $r \in \cR$ and $c \le 2$. For any given $\alp>0$,  Objectives in \eqref{eqn:OBJ-t} and \eqref{eqn:OBJ} for $\nadap(\alp)$ both converge with a rate upper bounded by (\ie as least as fast as) $\frac{4m\sum_{r \in \cR} w_r}{n^2 \exp(t/n^2)}$ and $ \frac{4m  \sum_{r \in \cR} w_r }{T}$, respectively.
\end{theorem}

We first show the following useful lemma\footnote{This can be viewed as a special case of Theorem \ref{thm:conv} in Appendix~\ref{sec:sup-basic}.}. Consider the Markov Chain $\sM(\nadap(\alp))$ as defined in \eqref{eqn:trans-mat-1} and let $\tau(\ep)$ be the mixing time of $\sM(\nadap(\alp))$ as defined in Section~\ref{sec:prelims}.
\vspace{0.1in}
\begin{lemma} \label{lem:mix}
$\tau(\ep) \le n^2 \ln(2m/\ep)$ when the arrival distribution of online requests is uniform such that $p_r=n^{-2}$ for all $r \in \cR$ and $c \le 2$.  
\end{lemma}

We defer the full proof of  Lemma \ref{lem:mix} to Appendix~\ref{sup:main}. Now we show the full proof of Theorem \ref{thm:error}.

\begin{proof}
Setting $\tau(\ep)=t$ implies that $\ep \geq 2m \exp(-t/n^2)$ in Lemma \ref{lem:mix}. Let $\ep = 2m \exp(-t/n^2)$. Recall that the definition of $\tau(\ep)=\min\{t: d(t') \le \ep,\forall t' \ge t\}$ where $d(t)=\max_{\x \in \Ome} ||P^t(\x, \cdot)-\pi^*||_{TV}$. Thus,  we have

$$\max_{\x \in \Ome} ||P^t(\x, \cdot)-\pi^*||_{TV} \le 2m \big(\exp(-1/n^2) \big)^t.$$

Thus, we have that \emph{when the arrival distribution of online requests is uniform such that $p_r=n^{-2}$ for all $r \in \cR$ and $c \le 2$, the scalars in Theorem~\ref{thm:conv} are explicit, namely, $C=2m$ and $\beta=\exp(-1/n^2)$}. By Lemma \ref{lem:thm-sim}, we have
	\[
		|\gam_{u,v}^{(t)}-\gam_{u,v}| \le 2C \beta^t =4m \exp(-t/n^2).
	\]
Recall that $W(\alp)= \lim_{t \rightarrow \infty} W(\nadap(\alp), t)$ and $\Delta_\alp(t)=| W(\nadap(\alp), t) - W(\alp) |$. Thus the RHS in Equation~\eqref{ineq:proof-thm1-b} can be upper-bounded by $\frac{4m\exp(-t/n^2)}{n^2} \sum_{r \in \cR} w_r$. Similarly, we have that 
	\[
			\widehat{\Delta}_\alp(T) \le \frac{1}{T}\sum_{t=1}^{T} \Delta_\alp (t) \le \frac{4m  \sum_{r \in \cR} w_r }{ T \cdot n^2 } \frac{1}{1-\exp(-1/n^2)} \sim   \frac{4m  \sum_{r \in \cR} w_r }{T}. \qedhere
	\]
\end{proof}

\subsection{A Lower Bound on Convergence Rates} 
We show that even in very special cases, the convergence rate shown in Theorem~\ref{thm:error} is \emph{almost} optimal. When $c=1$, we have the following lower bound where the dependence on $t$ and $T$ nearly matches  our upper bounds.

\begin{theorem}[Lower bound for $\Theta(\beta^t)$]
\label{thm:lb}
There is an instance with uniform arrival distribution where $p_r=n^{-2}$ for all $r \in \cR$ and $c=1$ such that Objective \eqref{eqn:OBJ-t} for $\nadap(1)$ has an asymptotic convergence rate of $\frac{2m\sum_{r \in \cR} w_r}{n^3\exp(t/n)}$.
\end{theorem}

\begin{proof}
Consider the following instance. 
Let $p_r=1/n^2$ for all $r\in \cR$ with $|\cR|=n^2$ and $c=1$. Thus the state space $\Ome=\{\x=\{0,1\}^n: \sum_{u} x_u=m\}$. Suppose we index the $n$ locations of $\cU$ as $u=1,2,\ldots,n$ and we set the initial state $\x^{(0)}$ as $x^{(0)}_u=1$ for $1 \le u \le m$ and $x^{(0)}_u=0$ for $u>m$. In other words, we have that at $t=0$, the first $m$ locations has one driver each and the rest have no driver. Recall that from assumptions we have that $m\ll cn$ which simplifies to $m \ll n$. Set $w_r=1$ if $\bar{r}=(m+1,1)$ (with origin of $m+1$ and destination of $1$) and $w_r=0$ otherwise.  Note that by definition of $\nadap(\alp=1)$, we will check the availability of driver at location $u$ with probability $1$ when a request $r=(u,*)$ arrives. Since all request brings a profit of $0$ except $\bar{r}=(m+1,1)$, the only source of profits is when $\bar{r}=(m+1,1)$ arrives, the system has one drive at location $m+1$ but no drive at location $1$. Thus, the limiting performance in \eqref{eqn:OBJ-t} is $W(\nadap(1))=\frac{1}{n^2} \gam_{m+1,1}$, where $\gam_{m+1,1}$ refers to the probability that the system has one drive at location $m+1$ while no drive at location $1$ in the limiting distribution.

As before, from Lemma~\ref{lem:irre} we have that $\sM(1)$ admits a unique uniform distribution $\pi^*$ over $\Ome$. Thus
$$\gam_{m+1,1}=\sum_{\x: x_{m+1} \ge 1, x_1<1} \pi^*(\x)=\frac{{n-2 \choose m-1}}{{n \choose m}}=\frac{m}{n}\frac{n-m}{n-1}.$$

Consider the locations $m+1$ and $1$ and let $\gam_{m+1,1,t}=\Pr[X^{(t)}_{m+1} \ge 1, X^{(t)}_1<1]$, where $\X^{(t)}$ is the random state in Markov Chain $\sM(\nadap(\alp=1))$ as defined in \eqref{eqn:trans-mat-1} with initial state $\x^{(0)}$. Now we compute an explicit formula for $\gam_{m+1,1,t}$. 

Let $u^*=1$ and $v^*=m+1$. There are $4$ possible states for the number of drivers in $u$ and $v$: $\bs_1=(1,0), \bs_2=(1,1),\bs_3=(0,0)$ and $\bs_4=(0,1)$, where $\bs_1$ denotes that there is one driver at $u=1$ and no driver at $m+1$. We define a Markov chain $\sM(u^*,v^*)$ over these $4$ states as follows
$$P_{(u^*,v^*)}=\begin{pmatrix}
\frac{n^2-n+1}{n^2} &\frac{m-1}{n^2} &\frac{n-m-1}{n^2} &\frac{1}{n^2}\\
\frac{n-m}{n^2} &\frac{n^2-2(n-m)}{n^2} &0&\frac{n-m}{n^2}\\
\frac{m}{n^2} &0  &\frac{n^2-2m}{n^2} &\frac{m}{n^2}\\
\frac{1}{n^2} &\frac{m-1}{n^2} &\frac{n-m-1}{n^2} &\frac{n^2-n+1}{n^2}\\
\end{pmatrix}$$

Consider the value $P(\bs_1,\bs_2)$. From state $\bs_1$ to $\bs_2$, we have that this event occurs iff a request ends at $m+1$ but starts at some location with one driver. Since the system has $m$ drivers in total, we have that there are $m-1$ request types contributing to the transition from $\bs_1$ to $\bs_2$. Thus $P(\bs_1,\bs_2)=(m-1)/n^2$. Similarly we can get the other values.

Note that for $\sM(u^*, v^*)$, we start with state $\bs_1$ and therefore $\gam_{m+1,1,t}=P^t_{(u^*,v^*)}(\bs_1,\bs_4)$. We have that 

\begin{align}
|\gam_{m+1,1,t}-\gam_{m+1,1}|& =\Big(1-\frac{1}{n}\Big)^t \cdot \frac{2mn-n-2m^2}{n(n-2)}+
\Big(1-\frac{2}{n}+\frac{2}{n^2}\Big)^t \cdot \frac{(m-1)(n-m-1)}{(n-1)(n-2)}. \\
&  \ge \Big(1-\frac{1}{n}\Big)^t \cdot \Big(\frac{2m}{n}-\frac{1}{n}-\frac{2m^2}{n^2} \Big)+\Big(1-\frac{2}{n}+\frac{2}{n^2}\Big)^t  \cdot (\frac{m}{n}-\frac{1}{n})(1-\frac{m}{n}-\frac{1}{n}).
\end{align}

When $1 \ll m \ll n$ and $t \ge 10 n$ the above expression can be approximated by
$$|\gam_{m+1,1,t}-\gam_{m+1,1}| \sim \exp(-t/n)\frac{2m}{n}+\exp(-2t/n)\frac{m}{n} \sim  \exp(-t/n)\frac{2m}{n}.$$
 
Therefore we have,
$$\Delta_t=|W(\nadap(1))-W(\nadap(1),t)|=\frac{1}{n^2} |\gam_{m+1,1}-\gam_{m+1,1,t}| \sim \frac{2m \sum_{r\in \cR}w_r}{n^3 \exp(t/n)}. \qedhere$$
\end{proof}

\vspace{-0.4in}
\section{Experiments}
\label{Sec:Experiments}
\begin{figure}[!ht]%
	\centering
	 \subfloat{{\includegraphics[scale=0.3]{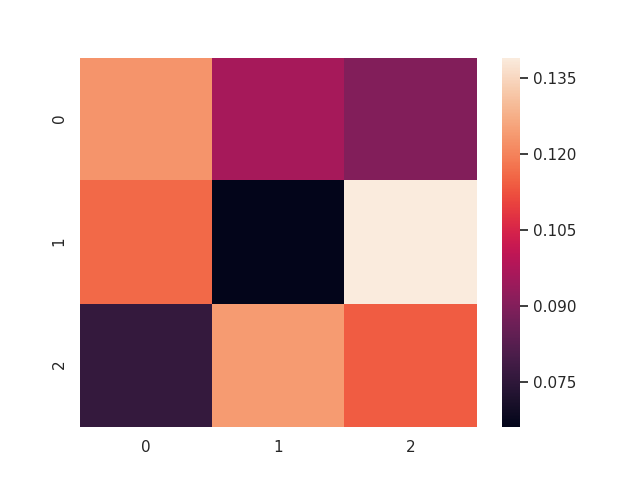} }}%
     \subfloat{ {\includegraphics[scale=0.3]{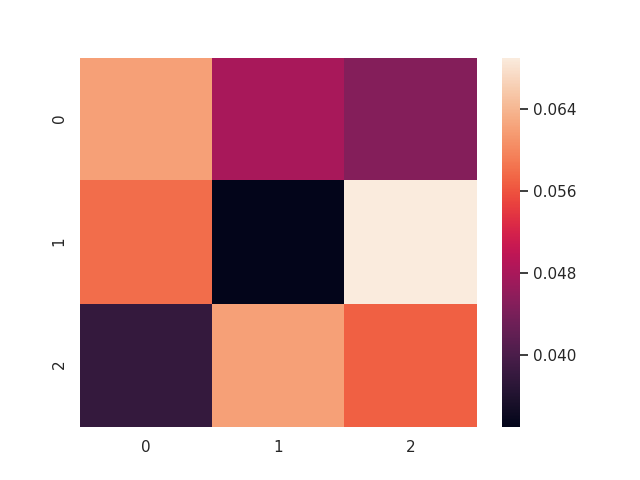} }}%
     \subfloat{ {\includegraphics[scale=0.3]{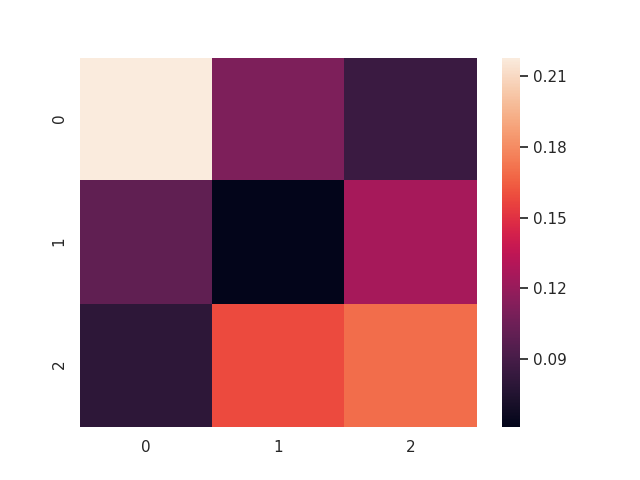} }}%
    \caption{3 $\times$ 3 grids; black to white represent low to high values. Plots 1, 2, 3 represent the distribution of requests with either starting or ending at a cell, requests starting at a cell and percentage of time spent by cars in a cell respectively by the optimal policy.}%
    \label{fig:heatmaps3}%
\end{figure}

In this section, we present experimental results on simulated synthetic data and a real-world dataset\footnote{\url{http://www.andresmh.com/nyctaxitrips/}}: the trip records for Jan 2013 of New York yellow cabs. We aim to experimentally study the convergence properties of objectives in Equations~\eqref{eqn:OBJ-t} and \eqref{eqn:OBJ} for various algorithms. We consider three kinds of algorithms; $\nadap(\alpha=0.8)$ and $\rand(\phi)$ described in the main section of the paper, where $\phi$ is chosen as the clock-wise permutation and the $\greedy$ algorithm whose action is defined as follows. When a request $r$ arrives at time $t$ with the system in state $\vec{X}^{(t)}$, we choose the neighbors such that the number of cars in state $\vec{X}^{(t)}$ are in decreasing order. Where appropriate, we also compare against a standard RL~\cite{sutton1998reinforcement} algorithm, namely, Deep-Q-Learning \cite{mnih_human-level_2015}, with agent and environment adapted from OpenAI Gym \cite{brockman_openai_2016} for 5 episodes of 2500 steps and experience replay.\footnote{\url{https://github.com/keon/deep-q-learning}} Fig.~\ref{fig:heatmaps3} represents exploratory analysis using optimal policy obtained via value iteration.\footnote{\url{https://www.ics.uci.edu/~dechter/publications/r42a-mdp_report.pdf}}

\xhdr{Preprocessing.} The records are pre-processed by restricting to trips whose starting and ending coordinates lie in the range of $(40.7014, 40.8024)$ for latitudes and $(-74.0041, -73.9552)$ for longitudes. We then split the records into $31$ parts with one part per date and further split them into Morning, Afternoon, Evening segments, corresponding to periods 7:00AM-11:00AM, 11:00AM-3:00PM, and 3:00PM-7:00PM, respectively. For brevity we only show results on the Morning segment; Appendix~\ref{sup:experiments} describes the other results. The Manhattan area is divided into a grid of $21 \times 11$. Using latitude-longitude coordinates, we map every record to a \emph{request} in this grid.%

\xhdr{Simulated experiments.} In the first experiment, we consider a grid of $5 \times 5$ cells and a total of $10$ cars with no accumulation limit on the number of cars in each location. We consider\footnote{We consider $T=10^4$ \emph{arrivals} within a 3-hour window. Additionally, the trips take multiple time-steps to complete.} $T=10^4$ and run $10^3$ independent experiments to compute the Obj.~\eqref{eqn:OBJ-t} and~\eqref{eqn:OBJ}. Each algorithm starts at the same \emph{adversarial} initial state (all $10$ cars are allocated in one single location) for every experiment. At each time step $t$, a request $r$ is sampled \emph{uniformly} at random. Figs.~\ref{fig:uniformSimulation1} and \ref{fig:uniformSimulation2} show the convergence of Obj.~\eqref{eqn:OBJ-t} and \eqref{eqn:OBJ} for the three algorithms. In particularly, we observe that the Obj.~\eqref{eqn:OBJ-t} and \eqref{eqn:OBJ} of $\nadap(0.8)$ both converge to the exact value as predicted in Theorem~\ref{thm:obj} (shown as the black line in the figure).

  \setlength{\belowcaptionskip}{-10pt}
\begin{figure*}[!h]
  \begin{minipage}{0.48\textwidth}
  	\centering
	\includegraphics[width=\linewidth]{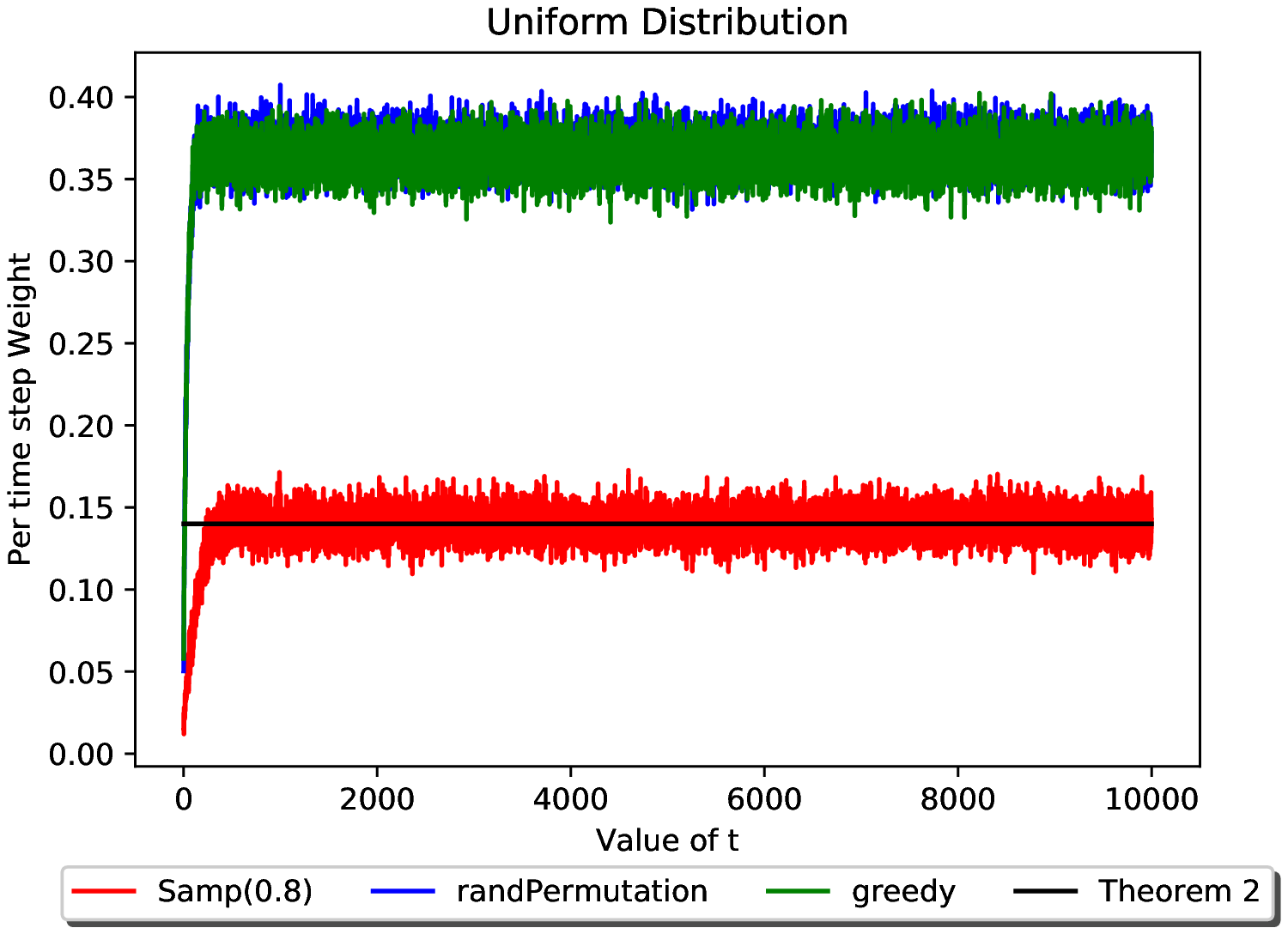}
  \caption{Obj.~\eqref{eqn:OBJ-t} on simulation with Uniform Arrivals}
  \label{fig:uniformSimulation1}
  \end{minipage}\hfill
  \begin{minipage}{0.48\textwidth}
  \includegraphics[width=\linewidth]{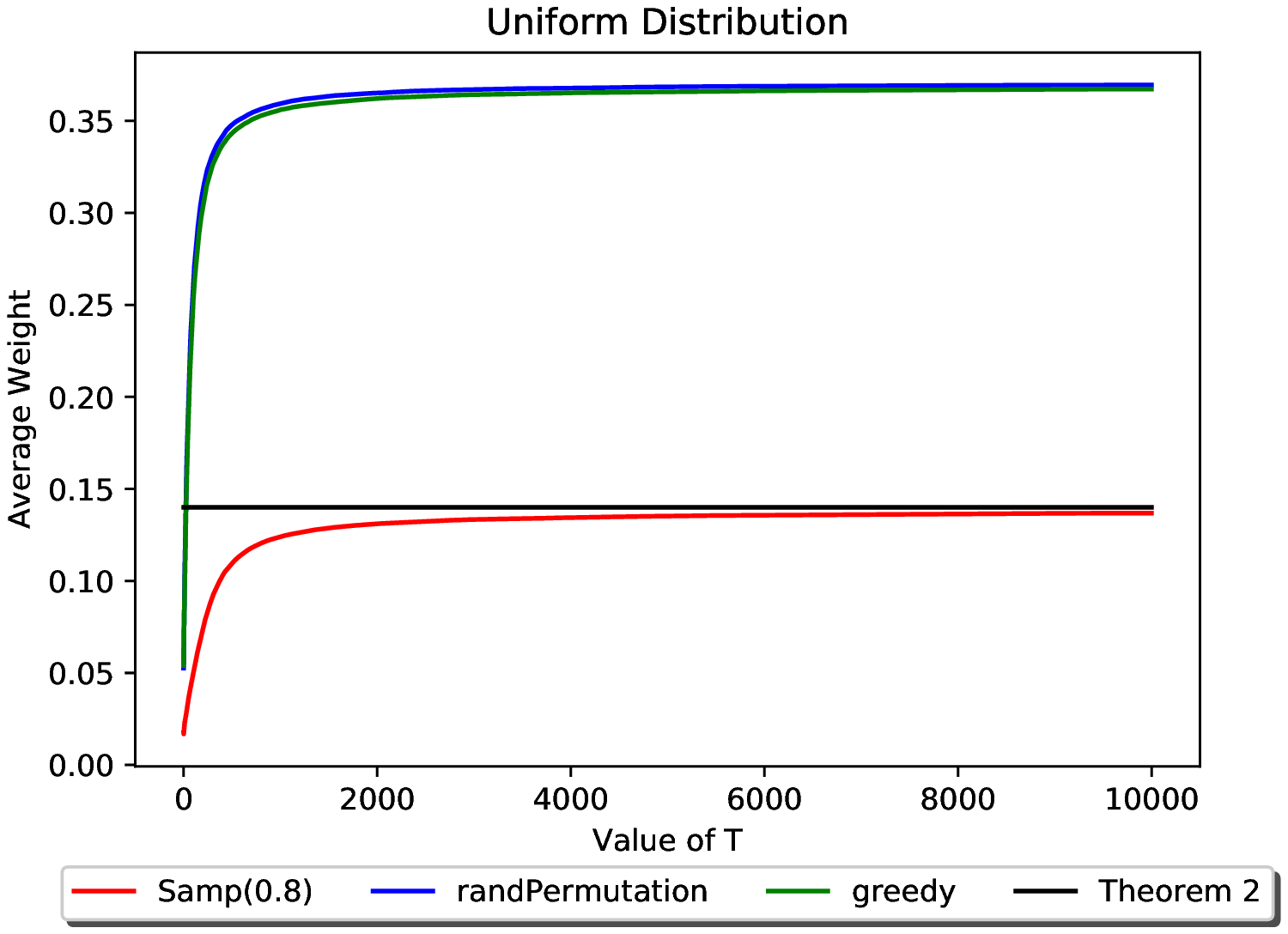}
  \caption{Obj.~\eqref{eqn:OBJ} on simulation with Uniform Arrivals}
  \label{fig:uniformSimulation2}
  \end{minipage}\hfill
  \end{figure*}
  
  \begin{figure*}[!h]
  \begin{minipage}{0.48\textwidth}%
  \includegraphics[width=\linewidth]{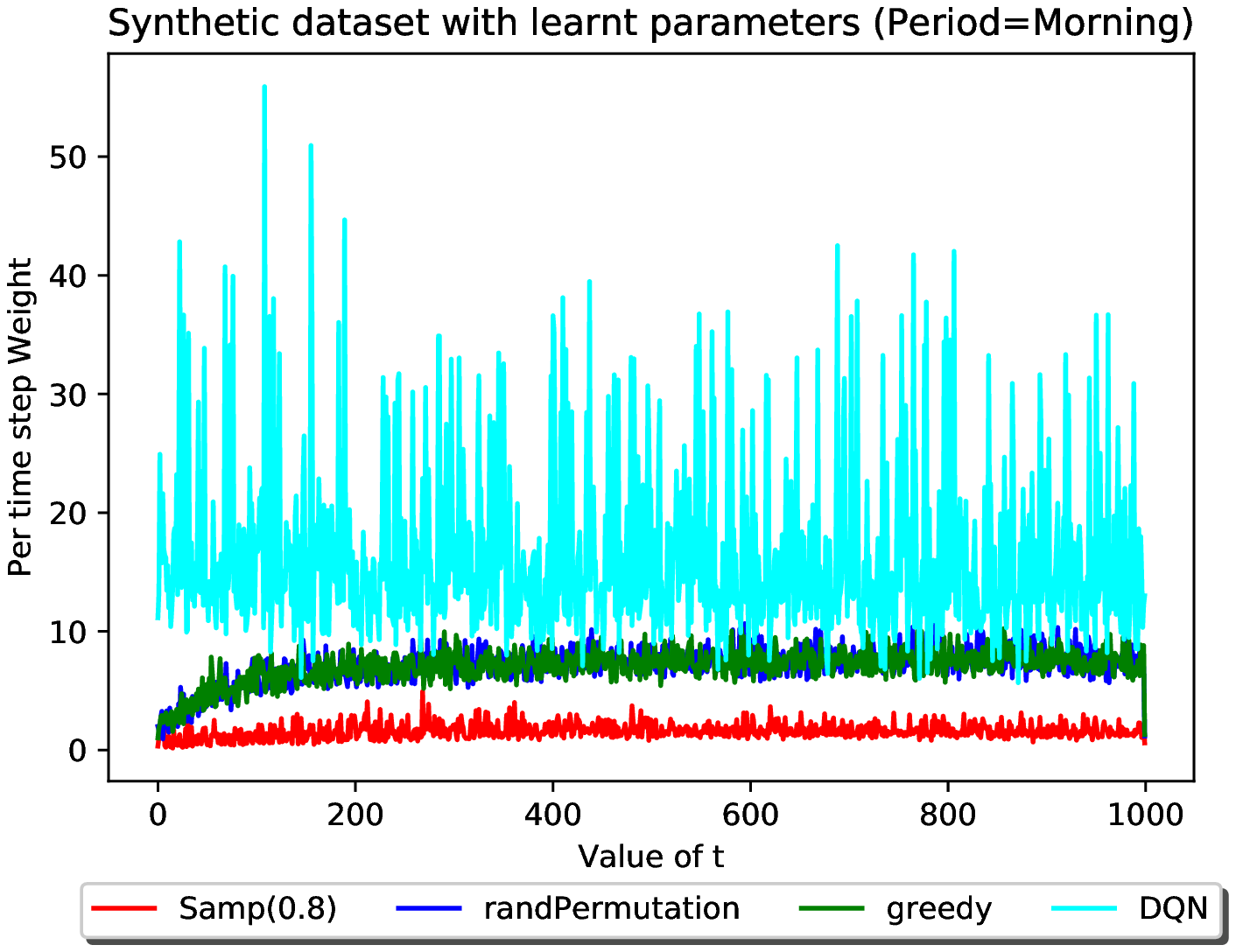}
  \caption{Obj.~\eqref{eqn:OBJ-t}, synthetic data, parameters learnt from real data}\label{fig:realSimulation1}
  \end{minipage}\hfill
  \begin{minipage}{0.48\textwidth}%
  \includegraphics[width=\linewidth]{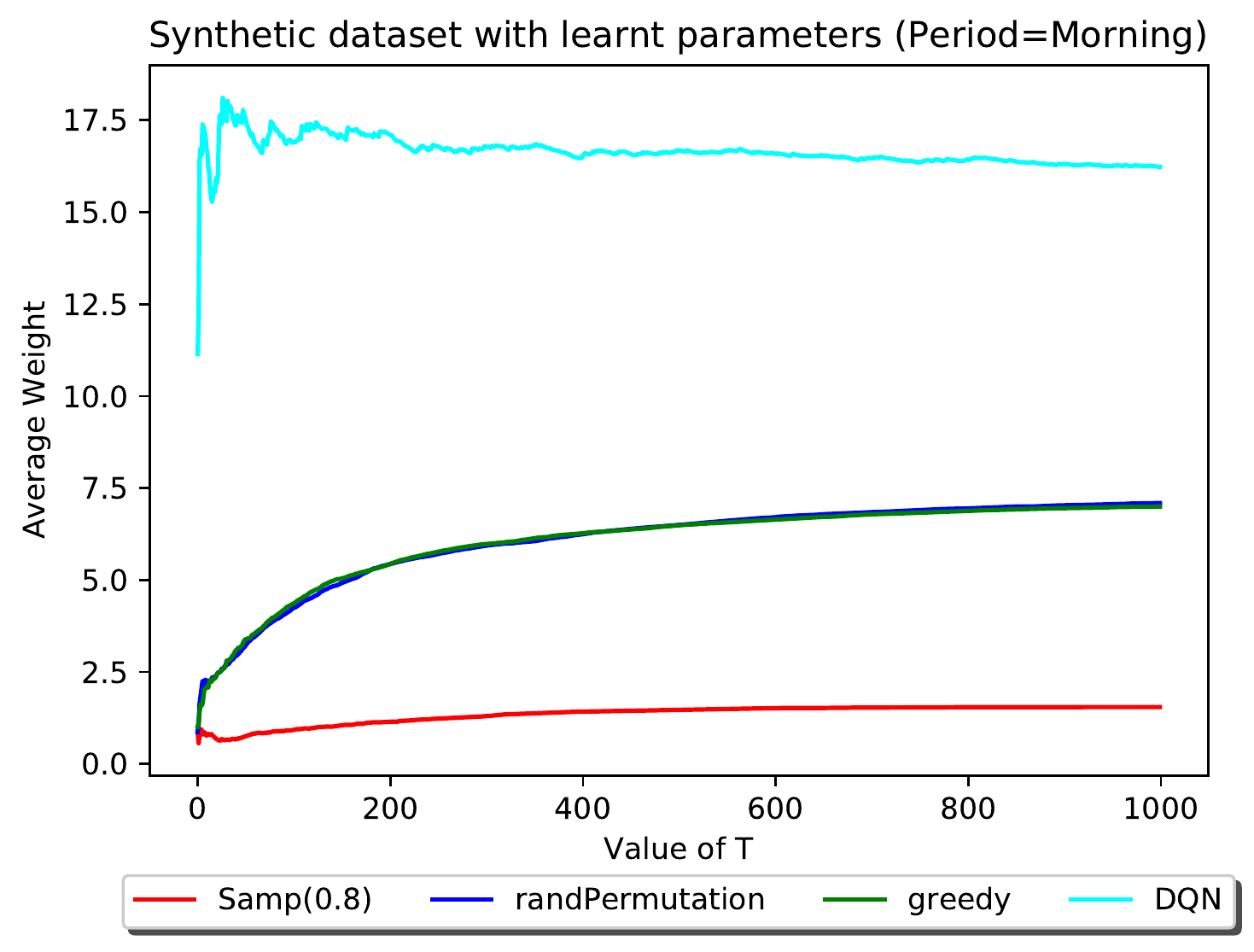}
  \caption{Obj.~\eqref{eqn:OBJ}, synthetic data, parameters learnt from real-data}\label{fig:realSimulation2}
  \end{minipage}
\end{figure*}

\xhdr{Synthetic experiments using learnt parameters.} Here, we ``learn'' the parameters from the processed taxi dataset, but \emph{simulate} the requests from these learnt parameters. In particular, we first learn the arrival rates of each of the requests using the data from Jan 2013 using a standard unbiased estimator. We then set the \emph{weight} for each request as the distance between the starting and ending location. At each time step $t$, a request $r$ is drawn from this learnt distribution. The grid size and the number of cars are identical to the real-world experiments (described below). Figs.~\ref{fig:realSimulation1} and \ref{fig:realSimulation2} show the convergence of  Obj.~\eqref{eqn:OBJ-t} and \eqref{eqn:OBJ} while Figs.~\ref{fig:error1} and \ref{fig:error2} show the corresponding absolute error bound. In Fig.~\ref{fig:error1} we fit a statistical best-fit exponential curve using least squares and in Fig.~\ref{fig:error2} we fit a statistical best-fit linear function. This corroborates our Main Theorem~\ref{thm:mainTheorem} as we see the same asymptotic behavior. Note that the parameters in Figs.~\ref{fig:realDataset} and \ref{fig:realSimulation2} are the same except the arrival assumption, thus the sharp difference does shed light on the importance of arrival assumption for convergence.

\begin{figure*}[!h]
  \begin{minipage}[t]{0.48\textwidth}
  \includegraphics[width=\linewidth]{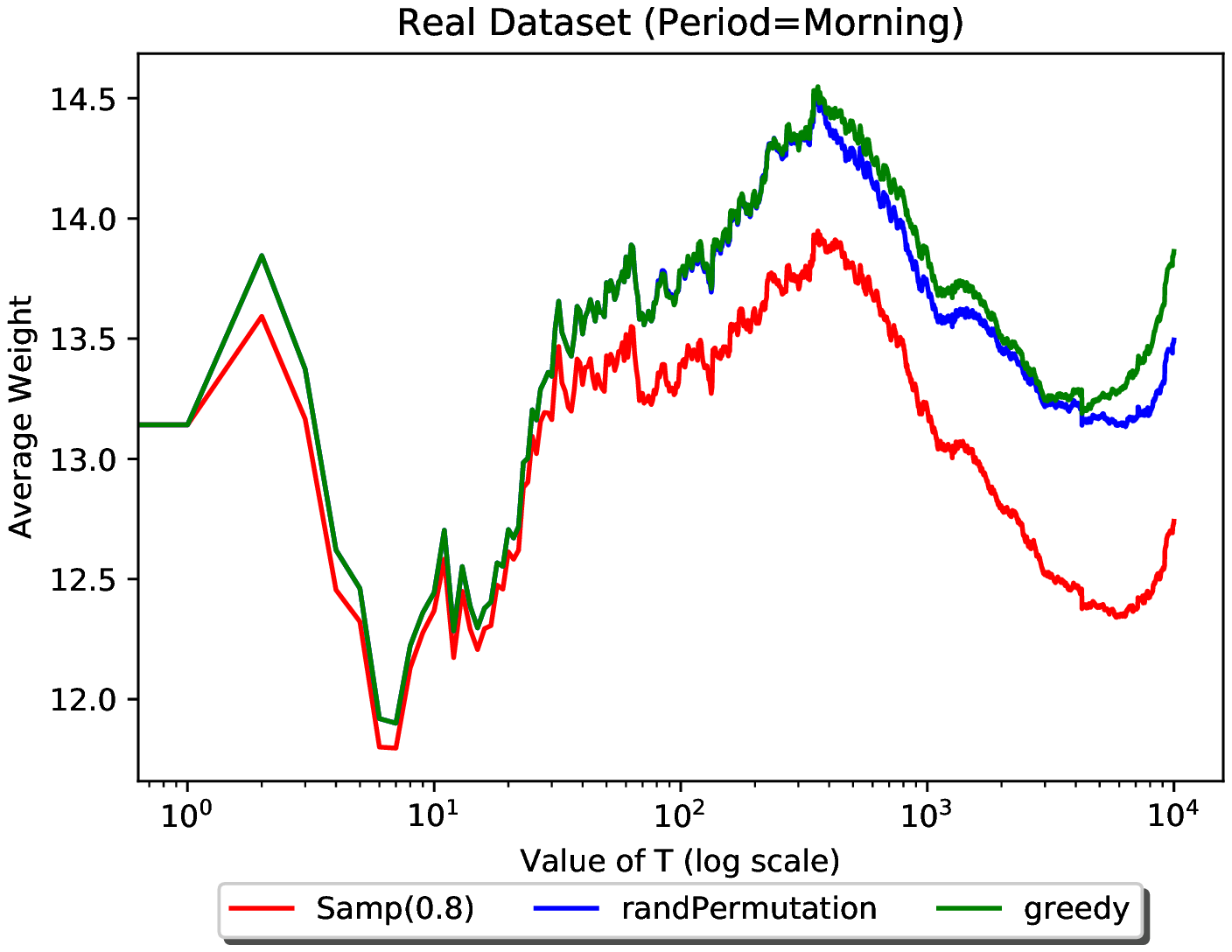}
	\caption{Obj.~\eqref{eqn:OBJ} on real dataset. $x$-axis in log-scale}
  \label{fig:realDataset}
  \end{minipage}\hfill
  \begin{minipage}[t]{0.48\textwidth}
  \includegraphics[width=\linewidth]{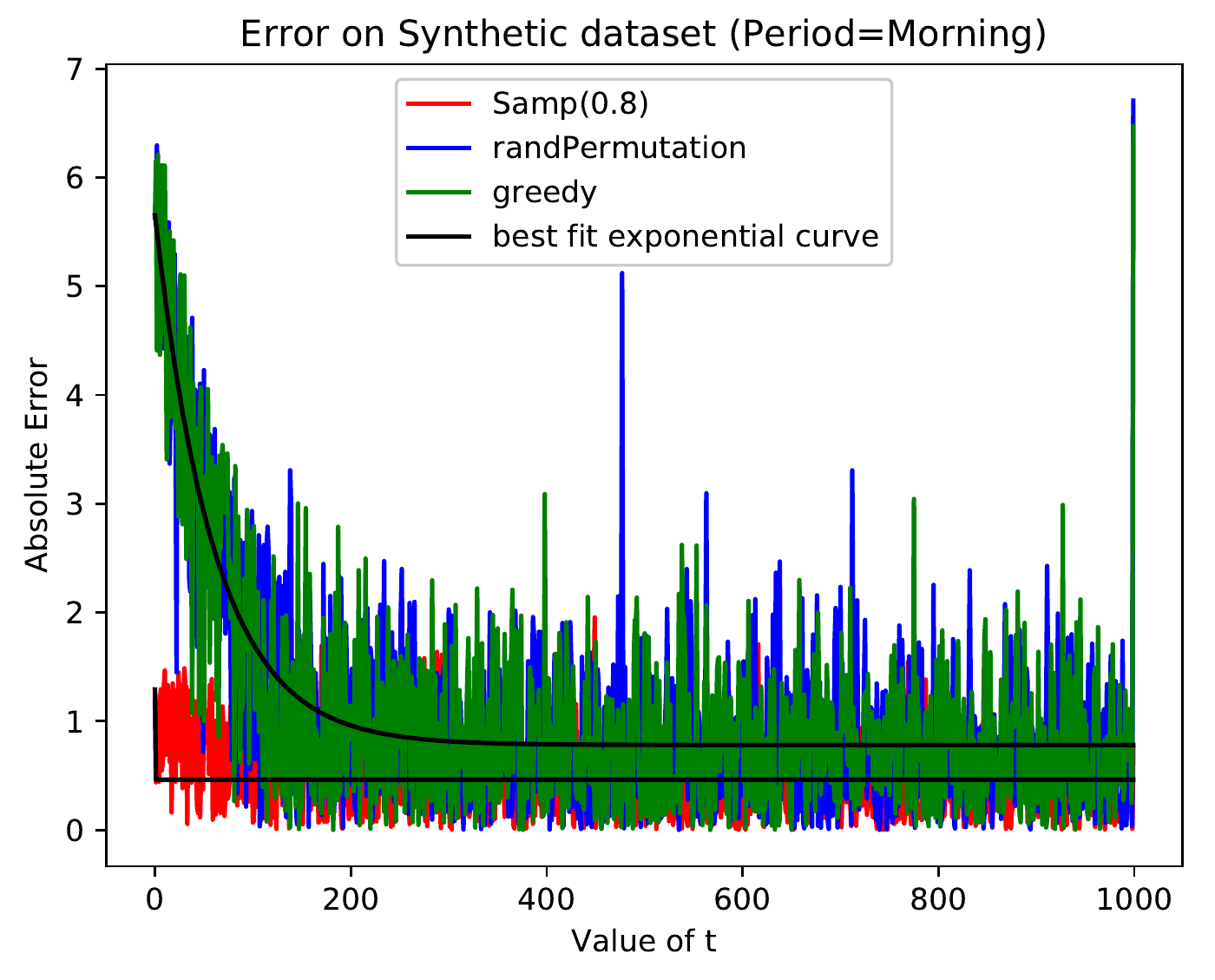}
  \caption{Error bound of Obj.~\eqref{eqn:OBJ-t} on synthetic data}
  \label{fig:error1}
  \end{minipage}\hfill
  \end{figure*}

  \begin{figure*}[!h]
  \begin{minipage}[t]{0.48\textwidth}%
  \includegraphics[width=\linewidth]{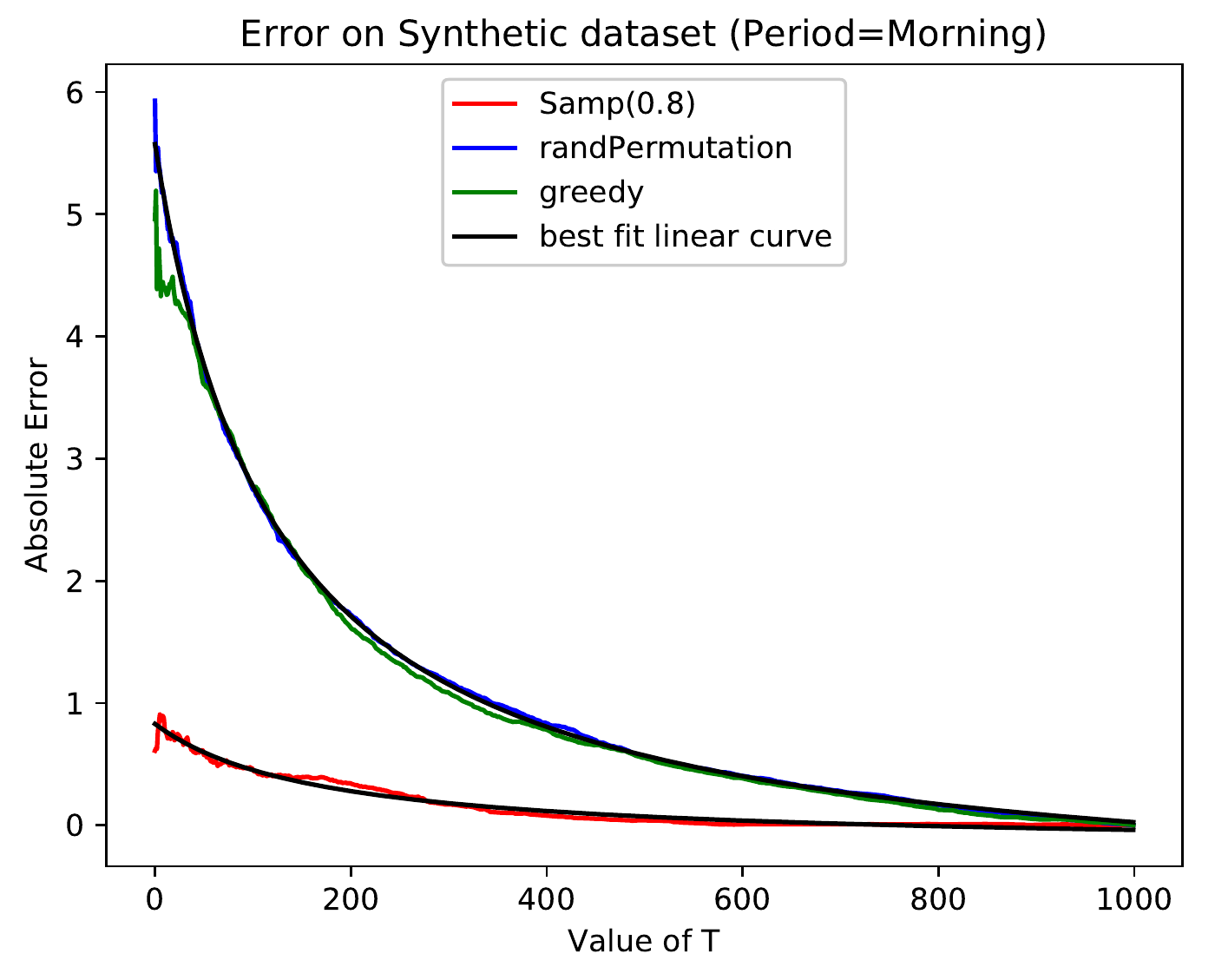}
  \caption{Error bound of Obj.~\eqref{eqn:OBJ} on synthetic data}
  \label{fig:error2}
  \end{minipage}\hfill
  \begin{minipage}[t]{0.48\textwidth}%
  \includegraphics[width=\linewidth]{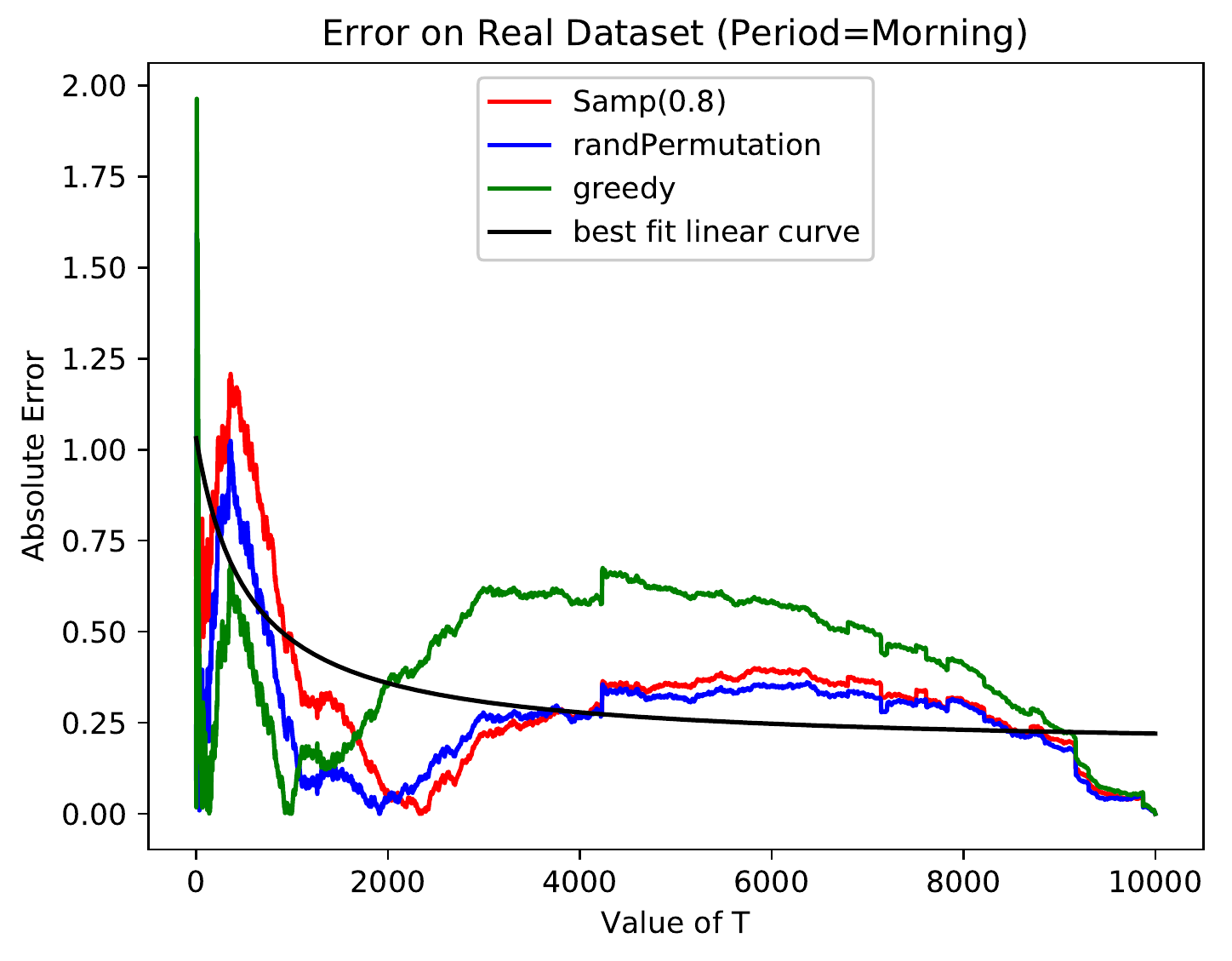}
  \caption{Error bound of Obj.~\eqref{eqn:OBJ} on real dataset}
  \label{fig:error3}
  \end{minipage}
\end{figure*}

\xhdr{Real-world experiments.} First, a random subset of $5000$ cars from the dataset is chosen and records corresponding to trips of these cars are sub-sampled. We use the real-time per-second arrival sequence of this sub-sampled dataset with $5000$ cars and a maximum congestion capacity of $50$ cars per cell as the experimental parameters. We study the convergence properties of Obj.~\eqref{eqn:OBJ} on this real-world dataset, for each period Morning, Afternoon and Evening. We use $T=10^4$ and use the $31$ days as ``independent'' runs to average our results over. Figs.~\ref{fig:realDataset} and \ref{fig:error3} show the evolution of Obj.~\eqref{eqn:OBJ} for the Morning period with x-axis having log-scale spacing and the corresponding error bound respectively. Fig.~\ref{fig:error3} uses the average of all $T$ rounds as the target convergence value.

\xhdr{Further discussion.} Our simulations show that even when starting in a pathologically bad state, all algorithms show the ``correct'' convergence behavior on both objectives (Figs.~\ref{fig:uniformSimulation1}, \ref{fig:uniformSimulation2}, \ref{fig:realSimulation1} and \ref{fig:realSimulation2}). In Fig.~\ref{fig:realSimulation1}, our state-dependent policy does almost as well as the standard RL algorithm while having \emph{good} convergence properties. DQN, on the other hand, wildly oscillates even after running for a large number of time-steps, thus supporting our main message on the importance of convergence. Additionally, algorithms proposed in this paper are \emph{simple} and extremely fast compared to DQN. On the real dataset, there are huge variations in the beginning steps and eventually the values do not vary a lot. However, they still do not fully converge like in the IID setting  (Figs.~\ref{fig:realDataset} and \ref{fig:error3}). This is mainly due to the fact that in practice the distribution is not \emph{truly} identical over all time steps and that full convergence is unlikely (\eg we provided such an example of non-i.i.d. arrivals in preliminaries). Nonetheless, the theory proved with the appropriate assumptions gives good insight into the performance in practice. In supplementary \ref{sup:experiments}, we show results for the Afternoon and Evening periods; the qualitative behaviors are the same. The fluctuations in the beginning also justify the importance of reaching a stationary distribution since platforms are \emph{risk-averse} and would prefer schemes with guaranteed throughput.

\vspace{-4mm}
\section{Conclusions \& Future Research}\label{sec:conclusion}
\vspace{-2mm}	
In this paper, we considered the rideshare dispatching problem and the corresponding Markov chain for the evolution of the matching process.  In turn, we characterized the driver distribution over time. In particular, we showed that the mixing time of this Markov chain is small, and gave explicit bounds on the convergence rate of various algorithms. Under practical assumptions, our theory shows that the convergence rates are fast, and in practice one would reach the stationary distribution in a short amount of time. To complement the theory we ran extensive experiments on both simulated and real datasets. Our theory gives accurate insight into rideshare dynamics in practice, and also complements a growing body of related research (e.g.,~\cite{ozkan2017dynamic,banerjee2017pricing,sid18}) that relies on assumptions that were, until now, underexplored and poorly understood. Still, relaxing any of the five commonly-made assumptions discussed in the preliminaries (\S\ref{sec:prelims}) would be practically useful.

\paragraph{Acknowledgments.} 
John Dickerson and Michael Curry were supported in part by NSF CAREER Award IIS-1846237 and DARPA SI3-CMD Award S4761. Aravind Srinivasan, Karthik Abinav Sankararaman and Pan Xu were supported in part by NSF CNS-1010789, CCF-1422569 and CCF-1749864, and by research awards from Adobe, Amazon and Google. Yuhao Wan was supported via an REU grant, NSF CCF-1852352, and was advised by John Dickerson.  John Dickerson and Aravind Srinivasan were both supported by a gift from Google and a seed grant from the Maryland Transportation Institute (MTI). Work done (and funding sources) when Karthik was at the University of Maryland, College Park.

\newpage	

\onecolumn

\appendix
\begin{center}{\LARGE\bf Appendix}\end{center}

\section{Primers on Markov Chains}
\label{sec:sup-basic}
In this section we provide a brief background on important concepts and results from Markov chains. For more details, please refer to these excellent monographs~(\cite{MCBook,guruswami2016rapidly}). Throughout, we denote a Markov Chain by $\sM$ with state space $\Omega$ and transition matrix $P$. Additionally, we assume that the number of states in $\sM$ is finite (\ie $|\Ome|< \infty$).

\xhdr{Irreducible Markov Chains.} $\sM$ is said irreducible if for any two states $\x, \y \in \Omega$, there exists an integer $t$ such that $P^t(\x,\y)>0$. In other words, it is possible to start at any state and reach any other state with a positive probability after finite number of transitions. 

\xhdr{Aperiodic Markov Chains.} For any state $\x$, let $\cT(\x) \doteq \{t \ge 1, P^t(\x,\x)>0\}$. The period of state $\x$ is defined to be the greatest common divisor of the entries in the set $\cT(x)$. A Markov Chain $\sM$ is said to be aperiodic if all states have period $1$.

Lemma 1.6 in \cite{MCBook} shows that all states in an irreducible chain have the same period. Thus to prove an irreducible chain $\sM$ is aperiodic, it suffices to show that there exists a state whose period is $1$.

For two distributions $\mu$ and $\nu$ on the state space $\Ome$, the \emph{total variation distance} between $\mu$ and $\nu$ is defined as $||\mu-\nu||_{TV}=\frac{1}{2} \sum_{\x \in \Ome} |\mu(\x)-\nu(\x)|$. 
\vspace{2mm}
\begin{definition}[Stationary Distribution]
	A distribution $\pi^*$ is said to be a \emph{stationary distribution} if $\pi^*=\pi^* P$. 
\end{definition}

\vspace{2mm}
\begin{definition}[Limiting Distribution]
	A distribution $\pi^*$ is said to be a \emph{limiting distribution} if for every initial state $\x$, we have $\lim_{t \rightarrow \infty}||P^{t}(\x, \cdot)-\pi^*||_{TV}=0$. Here $P^t(\x, \cdot)$ denotes that distribution over when after starting at $\x$ and running the process for $t$ steps. 
\end{definition}

\xhdr{Stationary and Limiting Distributions of Irreducible and Aperiodic Markov Chains.}
Any (finite-state) irreducible and aperiodic Markov Chain admits a unique stationary distribution. Additionally this is also the unique limiting distribution. More formally, we can prove the following convergence theorem.

\begin{theorem}\label{thm:conv}(Convergence Theorem 4.9 in \cite{MCBook}) 
Suppose a (finite-state) Markov Chain $\sM$ with transition matrix $P$ is irreducible and aperiodic with stationary distribution $\pi^*$. 
Then there exists a constant $C>0$ and $\beta \in (0,1)$ such that 
\[
		\max_{\x \in \Ome}||  P^t(\x, \cdot)-\pi^*||_{TV} \le C \beta^t.
\]
\end{theorem} 

Theorem~\ref{thm:conv} implies that the total variation distance between the distribution after $t$ steps and the unique stationary distribution exponentially decreases, irrespective of the initial state. \emph{Note that the terms $C$ and $\beta$ may depend on the parameters of the Markov Chain (\eg the size of $|\Ome|$) and is independent of $t$.}

\xhdr{When is the Distribution $\pi^*$ Uniform?} For any irreducible and aperiodic Markov Chain $\sM$, its unique stationary distribution $\pi^*$ is uniform over $\Ome$ iff in the transition matrix $P$, the sum of every column is equal to $1$. In particular, any symmetric matrix $P$ has all columns sum to $1$, since the sum of every row in $P$ is equal to $1$.

\xhdr{Rapid Mixing via coupling arguments.} Coupling and path coupling are common techniques used to prove that a Markov Chain is rapidly mixing. 
A (causal) coupling of the Markov Chain $\sM$ is a joint process $(\sX,\sY)=(\X^{(t)}, \Y^{(t)})$ on $\Ome \times \Ome$, such that each of the process $\sX$ and $\sY$ is a faithful copy of $\sM$, if considered marginally. In other words, we require that, for all possible $\x, \x', \y, \y' \in \Ome$, we have that 
$$\Pr[\ind{\X}{t+1}{}=\x'| \ind{\X}{t}{}=\x]=P(\x,\x'), ~~\Pr[\ind{\Y}{t+1}{}=\y'| \ind{\Y}{t}{}=\y]=P(\y,\y').$$

\begin{lemma}[Path Coupling Lemma 6.3 in \cite{guruswami2016rapidly}] \label{lem:path-coupling-a}
 Let $\bd$ be an integer-valued metric on $\Omega \times \Omega$ which has range over $\{0,1,\ldots, D\}$. Let $\cS$ be a subset of $\Omega \times \Omega$ such that for each $(\x,\y) \in \Omega \times \Omega$, there exists a path $\x=\z_0,\z_1,\ldots,\z_t=\y$ connecting $\x$ and $\y$ where $(\z_\ell,\z_{\ell+1}) \in \cS$ for all $0 \le \ell <t$ and $\bd(\x, \y)=\sum_{\ell=0}^{t-1} \bd(\z_\ell,\z_{\ell+1})$.  In other words, we just need to specify a subgraph $\cH$ with vertex set $\Omega$ and edge set $\cS$ and define a metric over $\cS$. Then $\bd(\x,\y)$ is simply the shortest path between $\x$ and $\y$ over $\cH$. Suppose we design a coupling $(\X,\Y) \rightarrow (\X',\Y')$ of the original Markov Chain $\sM$, which is defined over all pairs $(\X,\Y) \in \cS$ such that there exists a $\beta \in [0,1)$ such that $\E[\bd(\X',\Y')] \le \beta \bd(\X,\Y)$. Then the mixing time $\tau(\ep)$ of $\sM$ satisfies $\tau(\ep) \le \frac{\ln (D/\ep)}{1-\beta}$. 
\end{lemma}

\section{Missing Proofs in Main Sections} \label{sup:main}

\xhdr{Proof of Lemma~\ref{lem:irre}.}
		\begin{lemma}(Restated Lemma~\ref{lem:irre}).
			Under the hot-spot assumption, Markov chains $\sM(\alp)$  and $\sM(\phi)$ defined in \eqref{eqn:trans-mat-1} and \eqref{eqn:trans-mat-2} are both irreducible and aperiodic for any given $\alp>0$ and $\phi$.
		\end{lemma}
		
		It is well-known (see Theorem~\ref{thm:conv} in the primer) that a Markov chain that is both irreducible and aperiodic admits a unique limiting distribution $\pi^*$, which coincides with the unique stationary distribution.

		For this proof we refer to the transition matrix $P_\alp$ of $\sM(\alp)$ defined in \eqref{eqn:trans-mat-1} as $P$ for notation simplicity (\ie we drop the subscript $P_{\alp}$). 

\begin{proof}
\xhdr{$\sM(\alp)$ is irreducible.} From the definition of $P(\bf{\cdot})$ in \eqref{eqn:trans-mat-1}, we have that for any $(u,u')$-neighbor $(\x,\y)$, the inequality $P(\x,\y) \ge \alp p_{(u,u')}$ holds\footnote{A similar corresponding property holds for $\sM(\phi)$.}. Thus from the condition stated in Lemma \ref{lem:irre}, we get that $P(\x,\y)>0$ for any $(u,u')$-neighbor $(\x,\y)$ with either $u=u^*$ or $u'=u^*$.

In what follows, we assume that $u \neq u^*$ and $u' \neq u^*$.
Consider any two given states $\x, \y$ in $\Ome$. We show that there exists a positive integer $t > 0$ such that $P^t(\x,\y)>0$. This can be proved using the following cases. 

\begin{itemize}
\item \textbf{Case $\x \neq \y$.} W.l.o.g. assume that $(\x,\y)$ is a $(u,u')$-neighbor such that  $u \neq u^*$ and $u' \neq u^*$. Let $x_{u}=a, x_{u'}=b, x_{u^*}=d$. From the definition of
$(u,u')$-neighbor, we have that $y_u=a-1, y_{u'}=b+1$ and $y_{u^*}=d$.  The values on the ordered pairs of locations ($u$, $u'$, $u^*$) can be enumerated for various values of $d$.
\begin{enumerate} 

\item First case is when $d<c$, where $c$ is the capacity of the number of drivers in each location. Observe that there exists a feasible path from $\x$ and $\y$, namely,
	\[
		\x=(a, b,d) \rightarrow \x'=(a-1,b, d+1) \rightarrow \y=(a-1, b+1,d).
	\]
 Therefore we have that $P^2(\x,\y) >P(\x, \x') \cdot P(\x', \y)>0$ since $(\x,\x')$ is a $(u,u^*)$-neighbor while $(\x',\y)$ is a $(u^*,u')$-neighbor. 
 
 \item The other case is when $d=c$. Then there exists another feasible path connecting $\x$ and $\y$, namely,
\[
	\x=(a, b,c) \rightarrow \x'=(a,b+1, c-1) \rightarrow \y=(a-1, b+1,c).
\]
 Once again we have that $P^2(\x,\y)>0$.
 \end{enumerate}
 	Note that in both these cases we critically used the \emph{hot-spot} assumption.

\item \textbf{Case $\x=\y$.} W.l.o.g. assume that $x_u=a$ with $a<c$ for some $u \neq u^*$ and $x_{u^*}=b$. In this case we enumerate the various values for the ordered pair $(u, u^*)$.

\begin{enumerate}
\item ($b=0$). 
	We consider a few different scenarios in this case. The easiest case is when $a \neq 0$. In this case we have the following non-zero path from $\x$ to $\x$ namely,
		\[
			\x=(a, b) \rightarrow \x'=(a-1,b+1) \rightarrow \x=(a, b).
		\]
	
	Using the \emph{hot-spot} assumption, we have that  $P^2(\x,\x)>P(\x,\x') \cdot P(\x',\x)>0$.

	The second case is when $a = 0$ while at least one of $v \in \mathcal{N}(u)$ has a driver. We have a non-zero path in this case, with the changes in the values of the ordered pair $(u, v, u^*)$ as follows.
		\[
			\x=(0, h \neq 0, b) \rightarrow \x'=(0, h-1, b+1) \rightarrow \x=(0, h, b).
		\] 
		
		Using the \emph{hot-spot} assumption, we have that  $P^2(\x,\x)>P(\x,\x') \cdot P(\x',\x)>0$.
		
	The last case is when no drivers exist in $u$ and its neighborhood $\mathcal{N}(u)$. In this case,  any algorithm that dispatches using the drivers in this location will not find a driver and hence $P(\x, \x) = 1 > 0$ for all algorithms.

\item ($0<b$). Then observe that there exists a feasible loop on $\x$ as follows (using the \emph{hot-spot} assumption).
\[
	\x=(a, b) \rightarrow \x'=(a+1,b-1) \rightarrow \x=(a, b).
\]
Therefore we have $P^2(\x,\x)>P(\x,\x') \cdot P(\x',\x)>0$.
 \end{enumerate}
 
 \end{itemize}
 
\xhdr{$\sM(\alp)$ is aperiodic.} From Lemma 1.6 in~\cite{MCBook}, it suffices to show that there exists a state with period $1$. We do so by considering two cases on the relationship between $m$ and $c$.

\begin{itemize}
\item \textbf{Case $m \ge c$.} Consider a given state $\x$ with $x_{u^*}=c$. We have that $P(\x,\x) \ge \sum_{u \neq u^*} p_{(u,u^*)}>0$.
The above inequality is valid for the following reason. When we are at state $\x$ with $x_{u^*}=c$, then for any online request type $r$ with ending location $u^*$, any algorithm has to reject $r$, and hence we remain in state $\x$ for the next round. 

\item \textbf{Case $m<c$.} In this case, we consider a state $\x$ with $x_{u^*}=m$ while $x_u=0$ for all $u \neq u^*$. Like before we have that $P(\x,\x) \ge \sum_{u \neq u^*} \alp p_{(u,u^*)}>0$. 
\end{itemize}
Thus for each of the two cases we can always identify a state $\x$ with $P(\x,\x)>0$, implying that $\x$ has a period of $1$. 
\end{proof}

\xhdr{Proof of Theorem~\ref{thm:obj}.}

\begin{proof}
When $c=m$, we have that $\Ome=\{\x \in \mathbb{Z}_{+}^n: \sum_u x_u=m, x_u \ge 0, \forall u\}$. Note that 
$$|\Omega|={m+n-1 \choose n-1}={m+n-1 \choose m}.$$

Let $\gam_{u}$ be the probability that location $u$ has  at least one driver in the limiting distribution. We have that
\begin{align*}
\gam_u &=\Pr[X_u^{(\infty)} \ge 1]= \sum_{\x: x_u \ge 1} \pi^*(\x). \\
&=1-\frac{|\{\x \in \Ome: x_u=0\}|}{|\Ome|}
=1-\frac{{m+n-2 \choose m}}{{m+n-1 \choose m}}=\frac{m}{n+m-1}.
\end{align*}
Thus we claim that $\gam_u=\frac{m}{n+m-1}$ for all $u$. Notice that in our case $c=m$, $\gam_{u,v}=\gam_u$ for all $u,v \in \cU$ since if there is a driver at some location $u$, it surely implies that all other locations have less than $c$ drivers. 
Therefore,
\begin{align}
\OBJ(\nadap(\alp))&=\sum_{r=(u,v) \in \cR} p_r \cdot w_r \cdot \Big(  \alp \gam_{u,v}+ \frac{1-\alp}{4} \sum_{k \in \cN(u)} \gam_{k,v} \Big) \label{eqn:unif-a} \\
&=\sum_{r=(u,v) \in \cR} p_r \cdot w_r \cdot \Big(  \alp \gam_{u}+ \frac{1-\alp}{4} \sum_{k \in \cN(u)} \gam_{k} \Big) \label{eqn:unif-c}\\
&=\sum_{r \in \cR} p_r \cdot w_r \cdot \frac{m}{n+m-1}=\frac{m \cdot p}{n+m-1} \sum_{r \in \cR} w_r. \label{eqn:unif-b}
\end{align}

Note that Equality \eqref{eqn:unif-b} assumes that each $u$ has $4$ neighbors with Manhattan distance $1$.
\end{proof}

\xhdr{Proof of Lemma \ref{lem:mix}}.
We use Lemma \ref{lem:path-coupling-a} to prove Lemma \ref{lem:mix}. When $p_r=1/n^2$  for all $r \in \cR$, we have that the transition matrix $P$ of $\sM(\alp)$ is reduced  to the following form: $P(\x,\y)=1/n^2$ iff $(\x,\y)$ is a $(u,u')$-neighbor for all possible $(u,u')$ with $u \neq u'$ and $P(\x,\y)=0$ if $\x \neq \y$ and $(\x,\y)$ is not a $(u,u')$-neighbor. For each given state $\x$, let $\cN(\x)$ be set of states $\y$ such that $(\x,\y)$ is a $(u,u')$-neighbor for some $(u,u')$. The Markov Chain $\sM(\alp)$ thus can be viewed as follows: we start with an initial state $\x_0 \in \Ome$; suppose at time $t$ we are at state $\X^{(t)}$; then $\X^{(t+1)}=\y$ with probability $1/n^2$ for each $\y \in \cN(\X^{(t)})$ and $\X^{(t+1)}=\X^{(t)}$ with the remaining probability. 

Define a natural metric $\bd$ over $\Ome \times \Ome$ as follows: $\bd(\x,\y)=\sum_u|x_u-y_u|$. Let $\cS$ be the set of all possible $(u,u')$-neighbors on $\Ome$. Observe that $\bd(\x,\y) \ge 2$ for any $\x \neq \y$ and $\bd(\x,\y)=2$ for all $(\x,\y)\in \cS$. We can verify that $\bd(\x,\y) =\sum_u|x_u-y_u| \le D=2m$ for all possible $(\x,\y)$. To apply Lemma~\ref{lem:path-coupling-a}, we design a natural coupling $(\X,\Y) \rightarrow (\X',\Y')$  over all possible $(\X,\Y) \in \cS$ such that $\X'$ and $\Y'$ each is a one-step transition from $\X$ and $\Y$ respectively according the transition matrix $P$. In other words, $\X'$ and $\Y'$ are random states each follows the respective distribution $P(\X, \cdot)$ and $P(\Y, \cdot)$. Therefore we claim that $\X'$ and $\Y'$ is a faithful copy of $\sM(\alp)$ with respective $P$. In the following, we prove that $\E[\bd(\X',\Y')] \le (1-\frac{1}{n^2}) \bd(\X,\Y)$. 

\begin{proof}
Consider the case $c=2$. Let $(\X,\Y) \in \cS$ is a $(a,b)$-neighbor such that $Y_{a}=X_a-1$ and $Y_b=X_b+1$. Consider the following alternate view of the transition of $\sM(\alp)$. Suppose we are at state $\X$; each round we sample a request $r=(u,u')$ with probability $1/n^2$; if $X_u \ge 1$ and $X_{u'} \le c-1$,  then set $X'_u=X_u-1$, $X'_{u'}=X_{u'}+1$ and $X'_{v}=X_v$ for all $v \notin\{u,u'\}$, such that $(\X, \X')$ forms a $(u,u')$-neighbor; otherwise set $\X'=\X$. Consider the following cases.

	\textbf{Case 1: $X_a=1, X_b=0$.} Consider an arbitrary request $r=(u,u')$. 
If $u=a$ and $u'=b$, then we have that $\X'=\Y$ and $\Y'=\Y$ and thus $\bd(\X',\Y')=0$, which occurs with probability $1/n^2$. Similarly, if $u=b$ and $u'=a$, we have $\X'=\X$ and $\Y'=\X$ and $\bd(\X',\Y')=0$, which occurs with probability $1/n^2$. We can verify that for all other request types, $\bd(\X',\Y') \le 2=\bd(\X,\Y)$. Thus, we claim that $\E[\bd(\X',\Y')] \le(1-\frac{2}{n^2}) \E[\bd(\X,\Y)]$. The details are summarized in the table below.

\begin{table}[h!]
\caption{$(\X,\Y) \in \cS$: $X_a=1, X_b=0$ and $Y_a=0, Y_b=1$.}
\begin{center}
 \begin{tabular}{cc }
   $r=(u,u')$ & $\bd(\X',\Y')$ \\  \hline
  $u \notin \{a,b\}$  &  2  \\ 
  $u=a, u' \neq b$ & 2 \\ 
       $u=a, u'=b$  & 0 \\ 
      $u=b, u' \neq a$ & 2 \\ 
            $u=b, u'=a$  & 0 \\ 
    \hline    
  \end{tabular}
\end{center}
\end{table}

	\textbf{Case 2: $X_a=1, X_b=1$.} Applying a similar analysis, we claim that $\E[\bd(\X',\Y')] \le(1-\frac{1}{n^2}) \E[\bd(\X,\Y)]$. The details are as follows.

\begin{table}[h]
\caption{$(\X,\Y) \in \cS$: $X_a=1, X_b=1$ and $Y_a=0, Y_b=2$.}
\begin{center}
 \begin{tabular}{cc }
   $r=(u,u')$ & $\bd(\X',\Y')$ \\  \hline
  $u \notin \{a,b\}$  &  2  \\ 
  $u=a, u' \neq b$ & 2 \\ 
   $u=a, u'=b$  & 0 \\ 
      $u=b$ & 2 \\
    \hline    
  \end{tabular}
\end{center}
\end{table}

	\textbf{Case 3: $X_a=2, Y_a=1$ and $X_b=1, Y_b=2$.} 
	
	Consider the following cases summarized in the table below. Thus we have $\E[\bd(\X',\Y')] \le(1-\frac{2}{n^2}) \E[\bd(\X,\Y)]$. 

\begin{table}[h!]
\caption{$(\X,\Y) \in \cS$: $X_a=2, X_b=1$ and $Y_a=1, Y_b=2$}
\begin{center}
 \begin{tabular}{cc }
   $r=(u,u')$ & $\bd(\X',\Y')$ \\  \hline
  $u \notin \{a,b\}$  &  2  \\ 
  $u=a, u' \neq b$ & 2 \\ 
      $u=a, u'=b$  & 0 \\    
      $u=b, u' \neq a$ & 2 \\ 
      $u=b, u'=a$ & 0\\ \hline    
  \end{tabular}
\end{center}
\end{table}

	\textbf{Case 4: $X_a=2, X_b=0$ and $Y_a=1, Y_b=1$.} Thus from the table below, we have $\E[\bd(\X',\Y')] \le(1-\frac{1}{n^2}) \E[\bd(\X,\Y)]$.

\begin{table}[h!]
\caption{$(\X,\Y) \in \cS$: $X_a=2, X_b=0$ and $Y_a=1, Y_b=1$}
\begin{center}
 \begin{tabular}{cc }
   $r=(u,u')$ & $\bd(\X',\Y')$ \\  \hline
  $u \notin \{a,b\}$  &  2  \\ 
  $u=a$ & 2 \\ 
      $u=b, u' \neq a$ & 2 \\ 
   $u=b, u'=a$  & 0 \\ \hline    
  \end{tabular}
\end{center}
\end{table}

From the case based analysis above, we have $\E[\bd(\X',\Y')] \le \beta \E[\bd(\X,\Y)]$ with $\beta=1-\frac{1}{n^2}$. By applying a similar analysis, we get the result for the case when $c=1$. Thus, by Lemma \ref{lem:path-coupling-a}, we have $\tau(\ep) \le n^2 \ln(2m/\ep) $.
\end{proof}

\section{Additional Experiments}\label{sup:experiments}

In this section we describe the results from the additional experiments. In particular, we conduct both simulated and real-world experiments similar to the main section for the periods Morning and Evening. Figures~\ref{fig:startingMorning}, \ref{fig:endingMorning} describes the request pattern for Afternoon while Figures~\ref{fig:startingEvening}, \ref{fig:endingEvening} show the patterns for Evening. Figures~\ref{fig:exp6Morning}, \ref{fig:exp7Morning} denote the results from using the Afternoon data for simulation while Figures \ref{fig:exp6Evening}, \ref{fig:exp7Evening} denote the results when the Evening data is used for simulation. Figures~\ref{fig:exp3Morning}, \ref{fig:exp3Evening} denotes the performance on the real-time requests for the Afternoon and Evening periods respectively. In general, all results follow the same pattern as those reported in the main section.

\begin{figure*}[!h]
	\minipage{0.48\textwidth}
	\centering
	 \includegraphics[width=\linewidth]{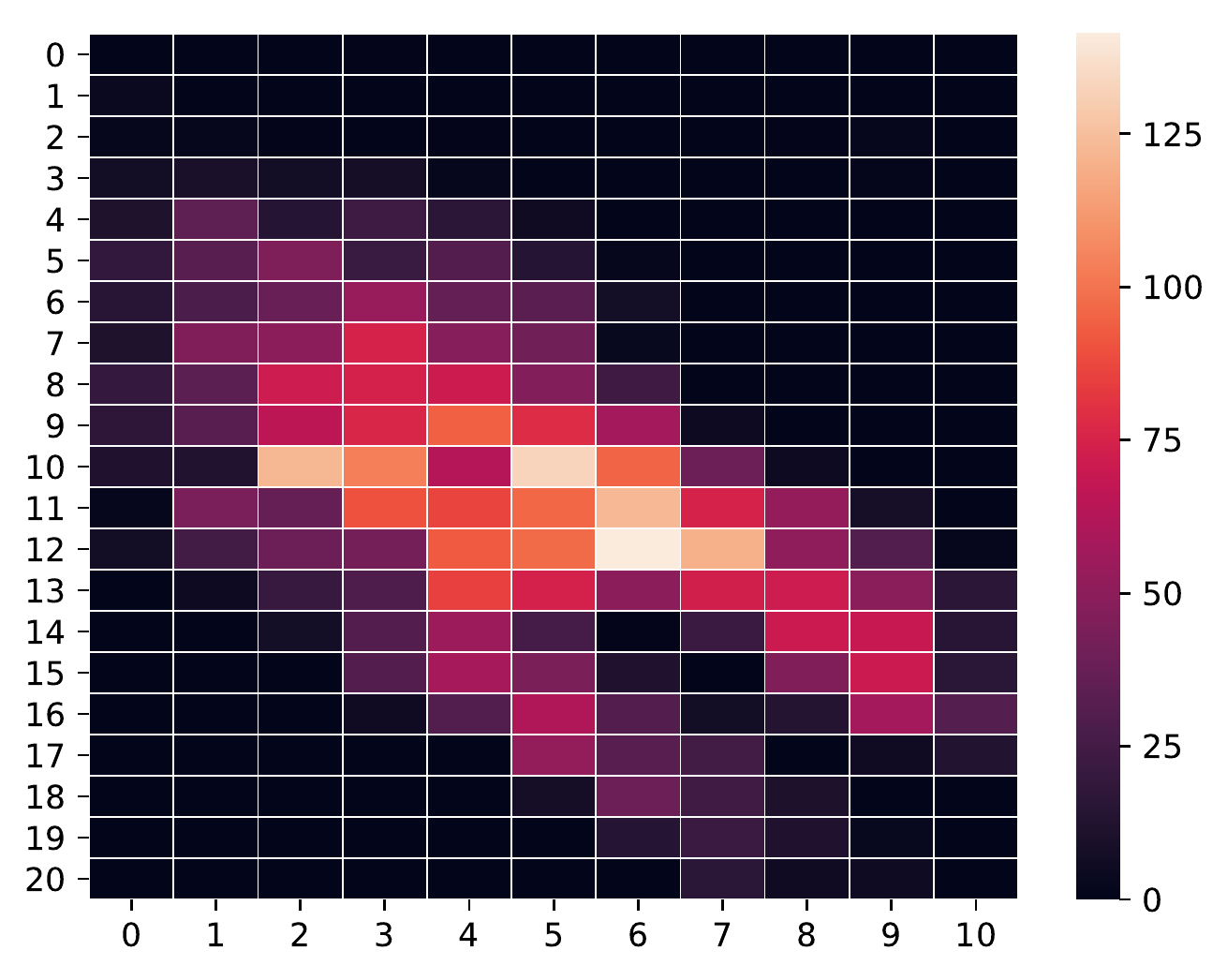}
	 \caption{Heat map of Starting Location of rides in NYC (11:00AM - 3:00PM)}
	\label{fig:startingMorning}
	\endminipage\hfill
	 \vspace{5mm}
	\minipage{0.48\textwidth}
	\centering
	\includegraphics[width=\linewidth]{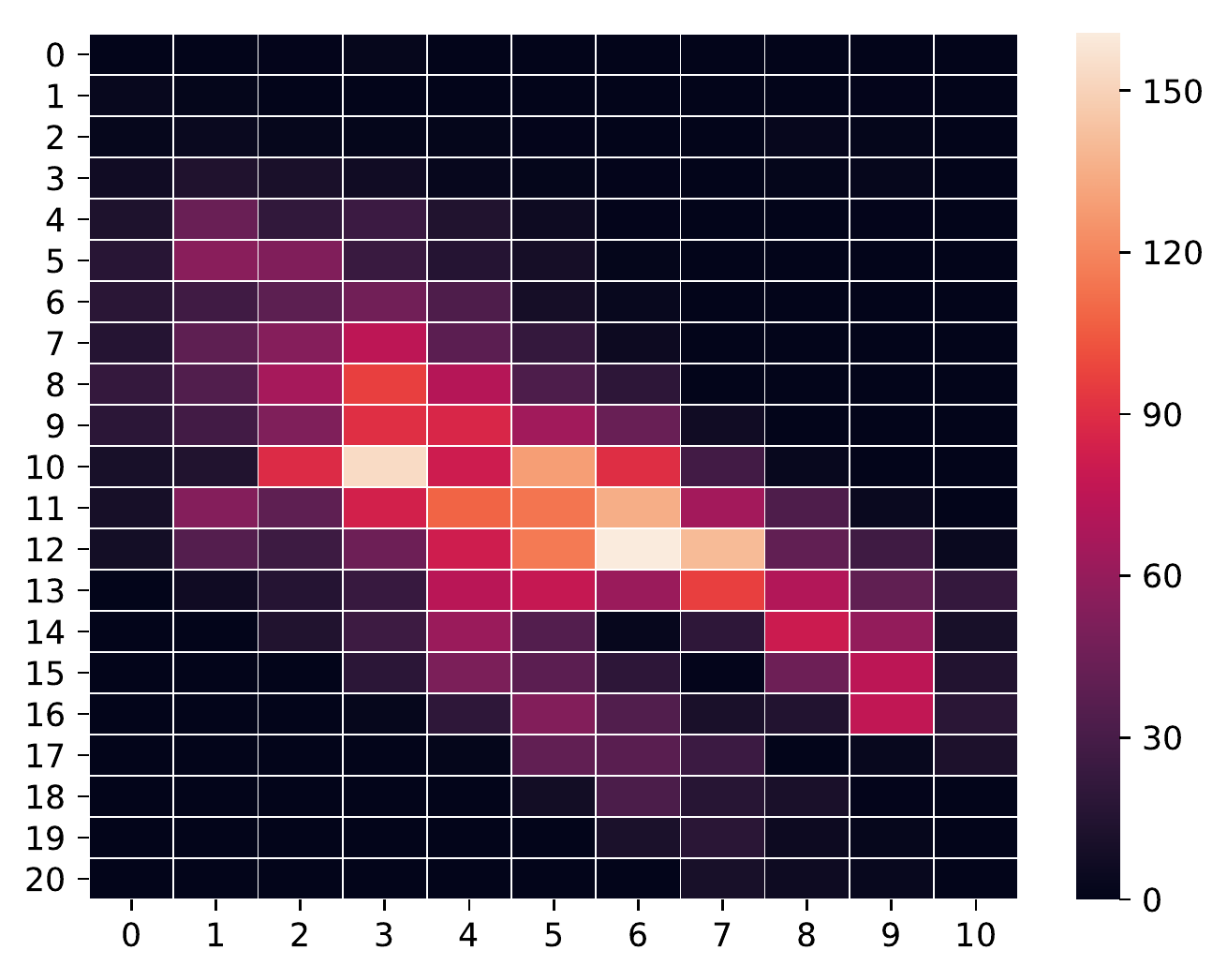}
	\caption{Heat map of Ending Location of rides in NYC (11:00AM - 3:00PM)}
	\label{fig:endingMorning}
	\endminipage
	 \vspace{5mm}
\end{figure*}

\begin{figure*}[!h]
	\minipage{0.48\textwidth}
	\centering
	 \includegraphics[width=\linewidth]{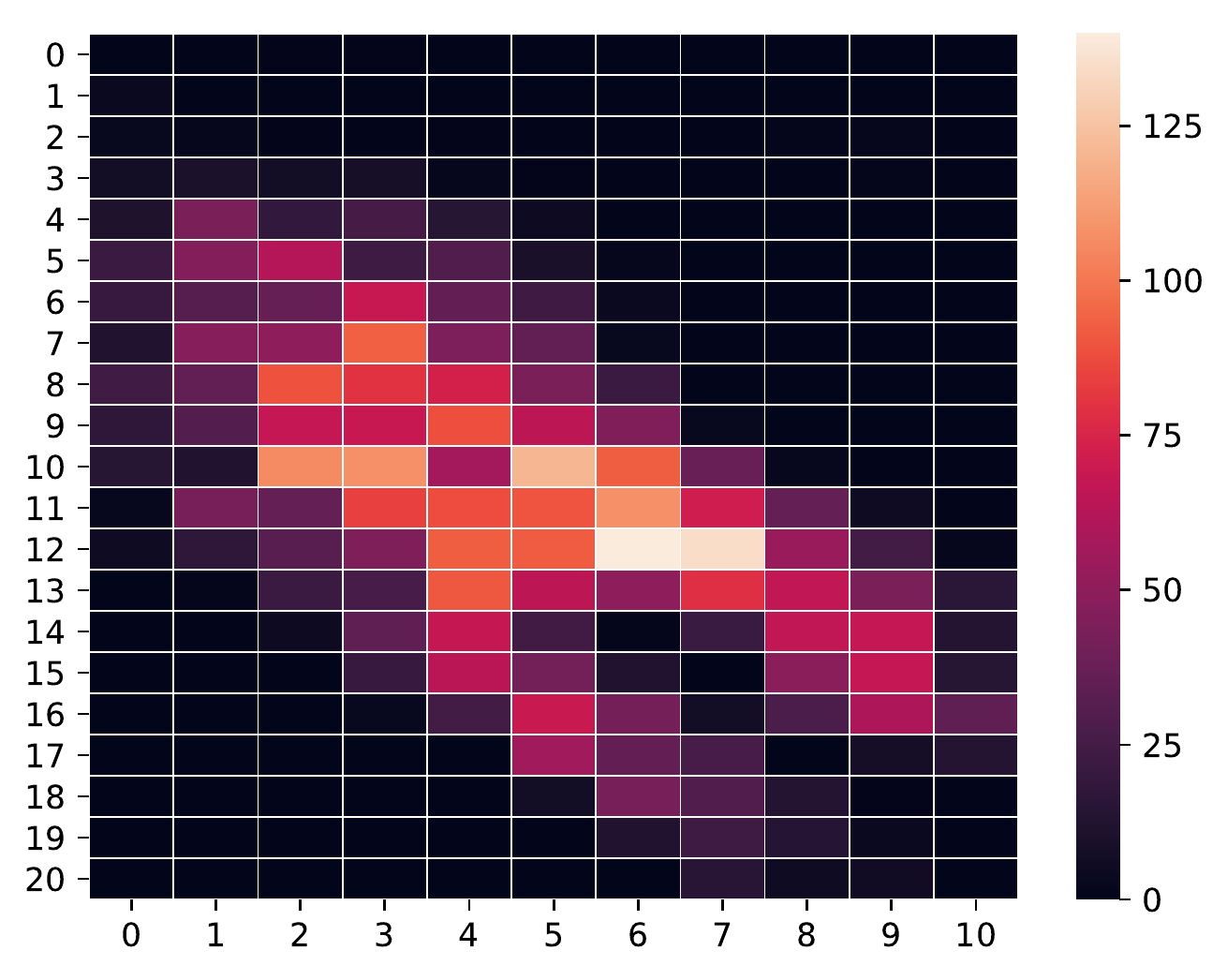}
	 \caption{Heat map of Starting Location of rides in NYC (3:00PM - 7:00PM)}
	\label{fig:startingEvening}
	\endminipage\hfill
	 \vspace{5mm}
	\minipage{0.48\textwidth}
	\centering
	\includegraphics[width=\linewidth]{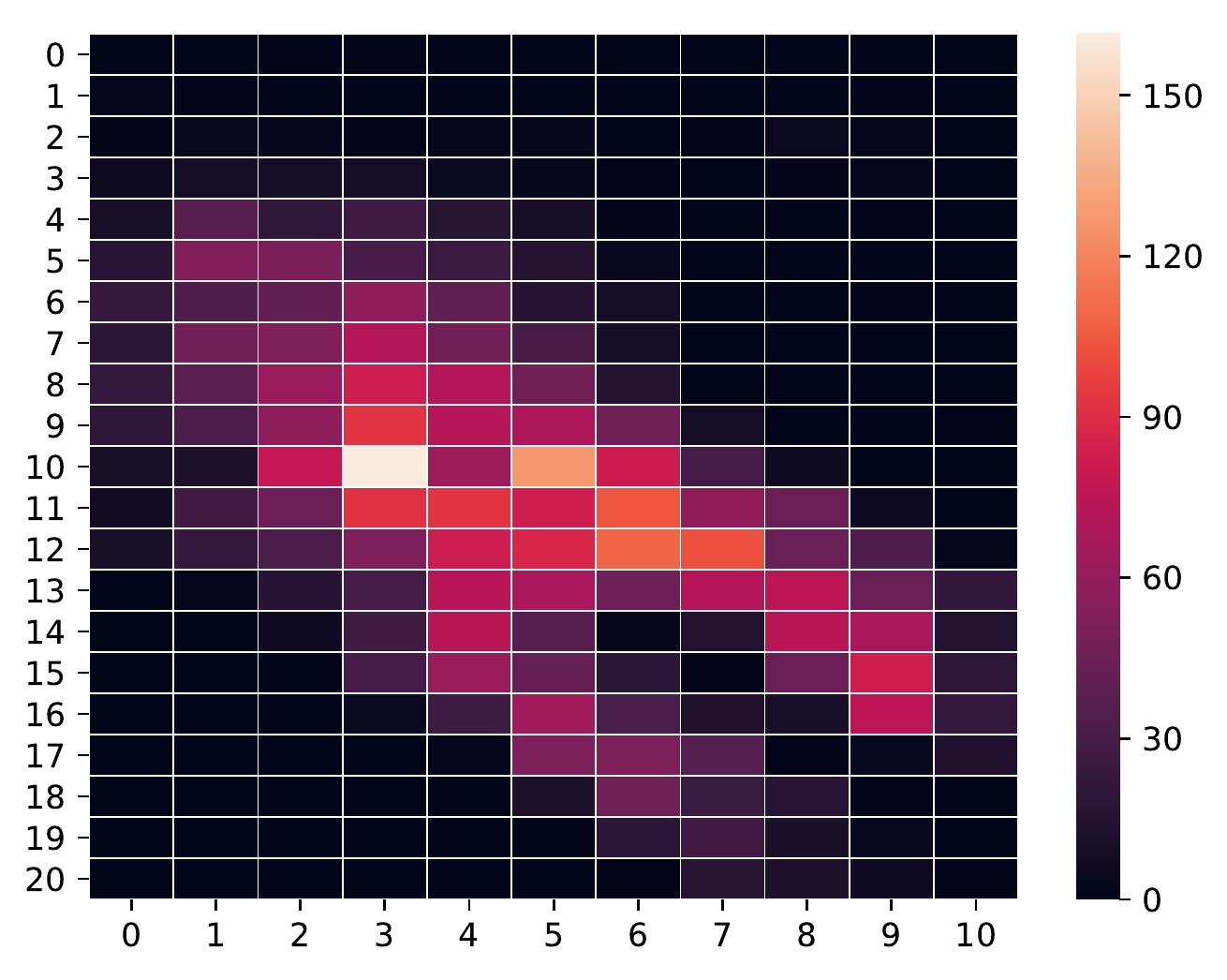}
	\caption{Heat map of Ending Location of rides in NYC (3:00PM - 7:00PM)}
	\label{fig:endingEvening}
	\endminipage
	 \vspace{5mm}
\end{figure*}

\begin{figure*}[!h]
	\minipage{0.48\textwidth}
	\centering
	 \includegraphics[width=\linewidth]{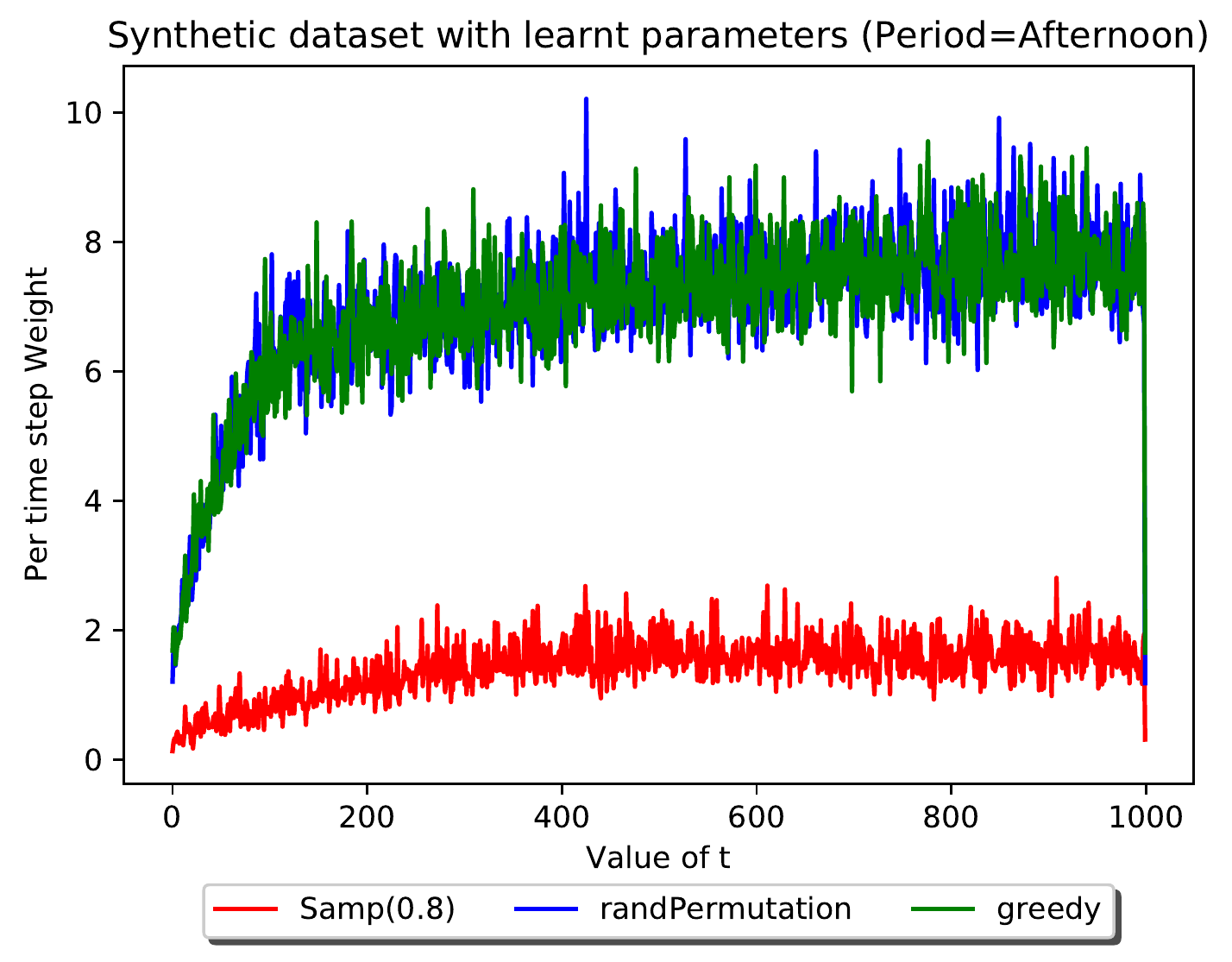}
	 \caption{Obj.~\eqref{eqn:OBJ-t} on simulation with real-data parameters (11:00AM-3:00PM)}
	\label{fig:exp6Morning}
	\endminipage\hfill
	 \vspace{5mm}
	\minipage{0.48\textwidth}
	\centering
	\includegraphics[width=\linewidth]{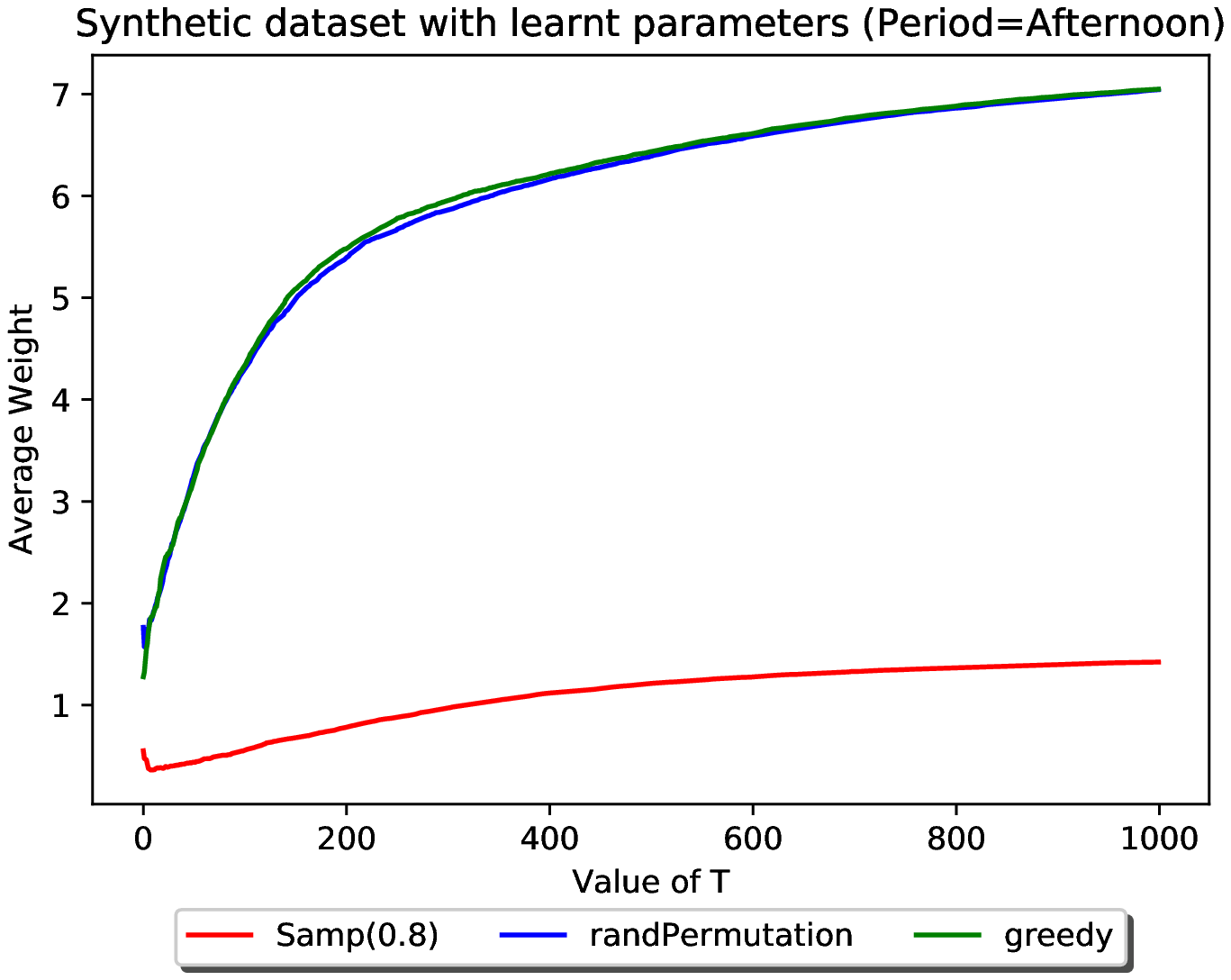}
	\caption{Obj.~\eqref{eqn:OBJ} on simulation with real-data parameters (11:00AM-3:00PM)}
	\label{fig:exp7Morning}
	\endminipage
	 \vspace{5mm}
\end{figure*}

\begin{figure*}[!h]
	\minipage{0.48\textwidth}
	\centering
	 \includegraphics[width=\linewidth]{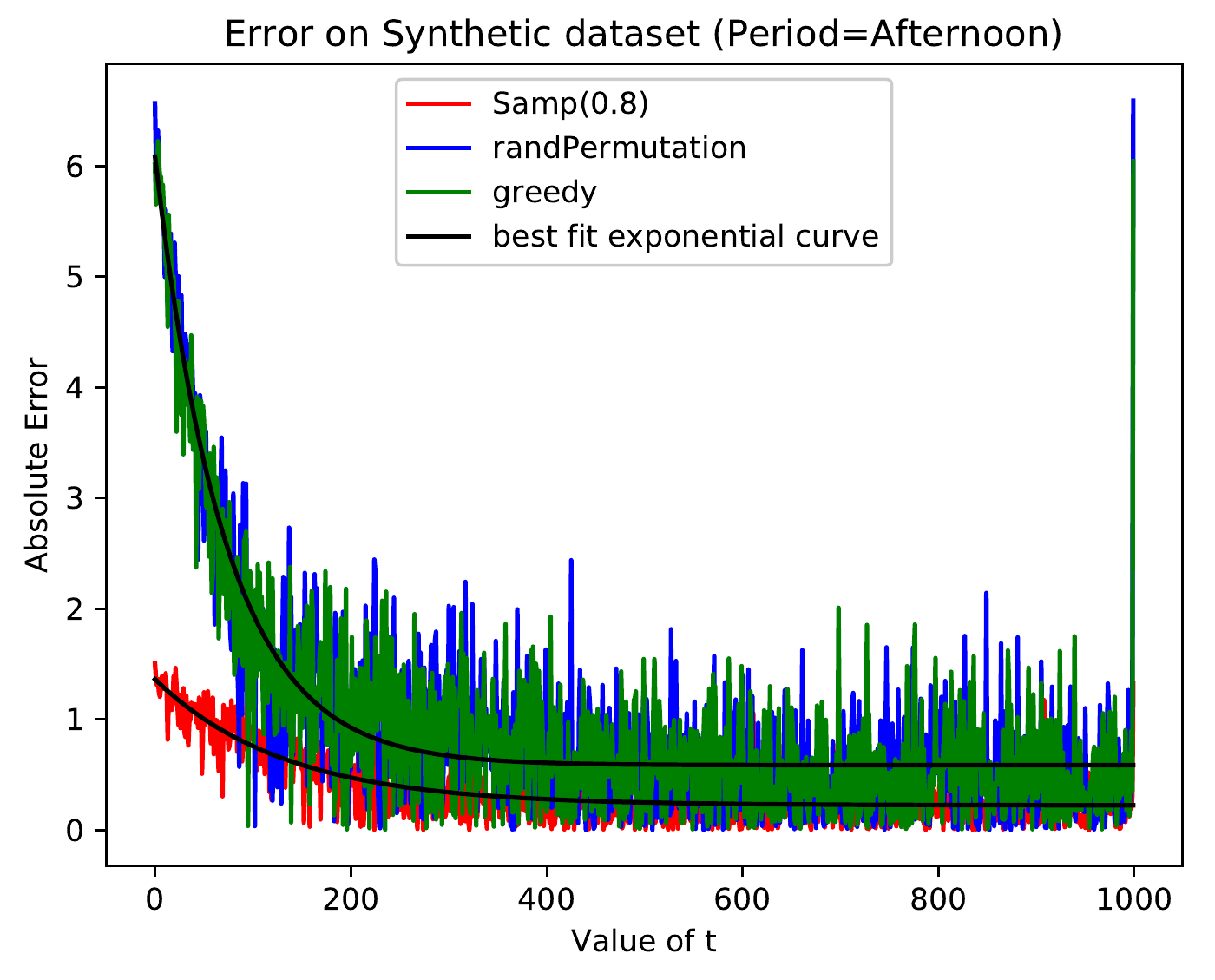}
	 \caption{Error in Obj.~\eqref{eqn:OBJ-t} on simulation with real-data parameters (11:00AM-3:00PM)}
	\label{fig:exp6MorningError}
	\endminipage\hfill
	 \vspace{5mm}
	\minipage{0.48\textwidth}
	\centering
	\includegraphics[width=\linewidth]{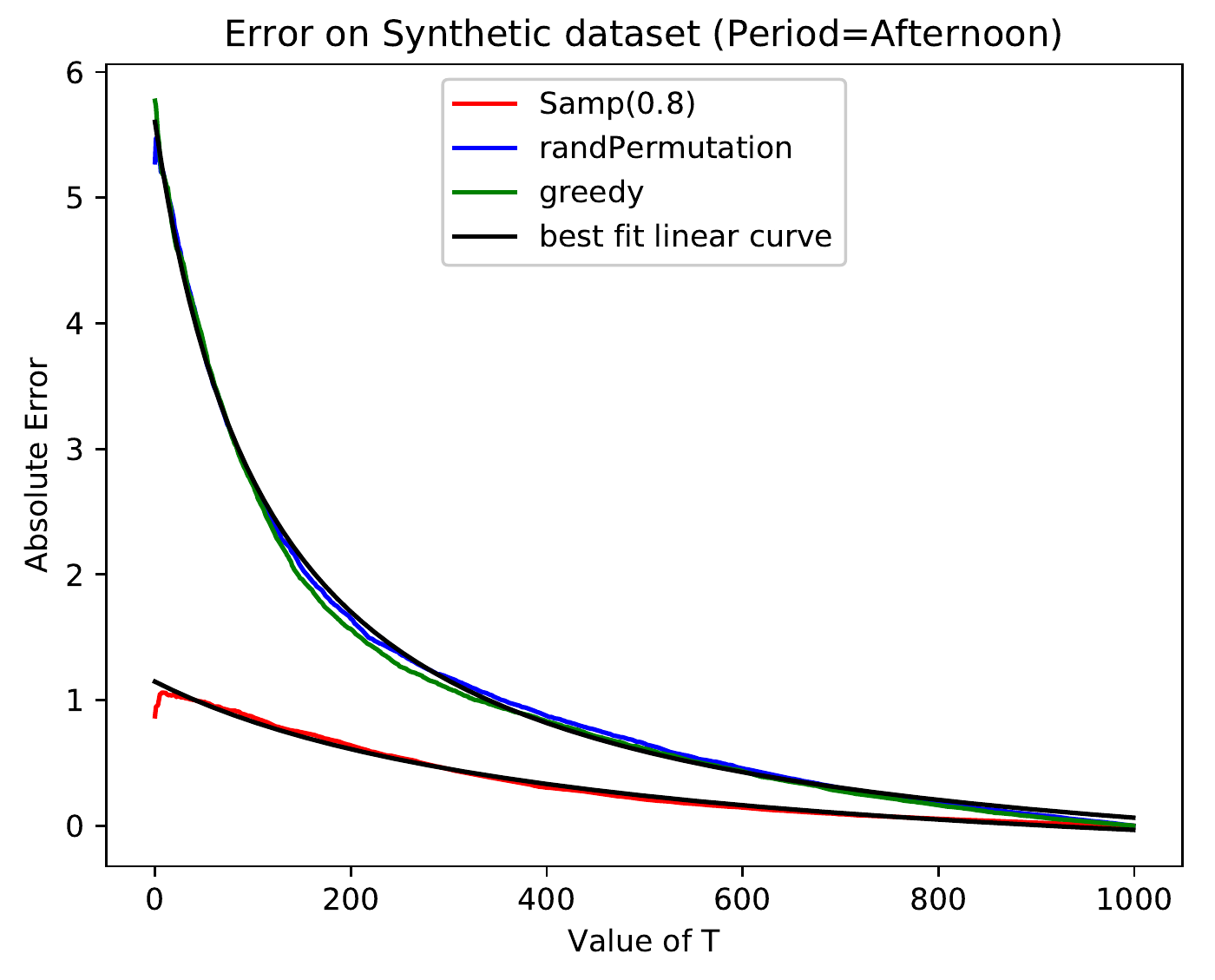}
	\caption{Error in Obj.~\eqref{eqn:OBJ} on simulation with real-data parameters (11:00AM-3:00PM)}
	\label{fig:exp7MorningError}
	\endminipage
	 \vspace{5mm}
\end{figure*}

\begin{figure*}[!h]
	\minipage{0.48\textwidth}
	\centering
	 \includegraphics[width=\linewidth]{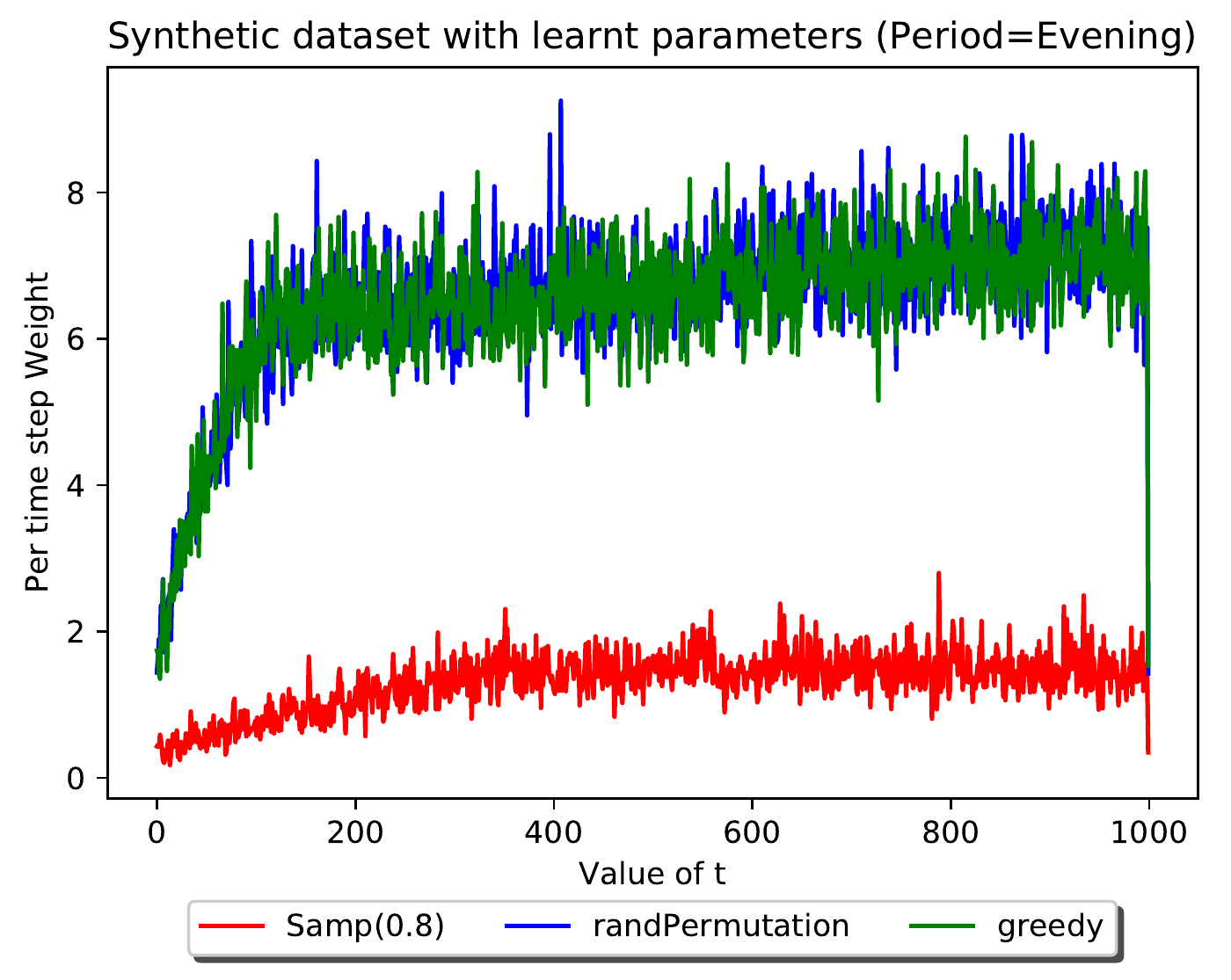}
	 \caption{Obj.~\eqref{eqn:OBJ-t} on simulation with real-data parameters (3:00PM-7:00PM)}
	\label{fig:exp6Evening}
	\endminipage\hfill
	 \vspace{5mm}
	\minipage{0.48\textwidth}
	\centering
	\includegraphics[width=\linewidth]{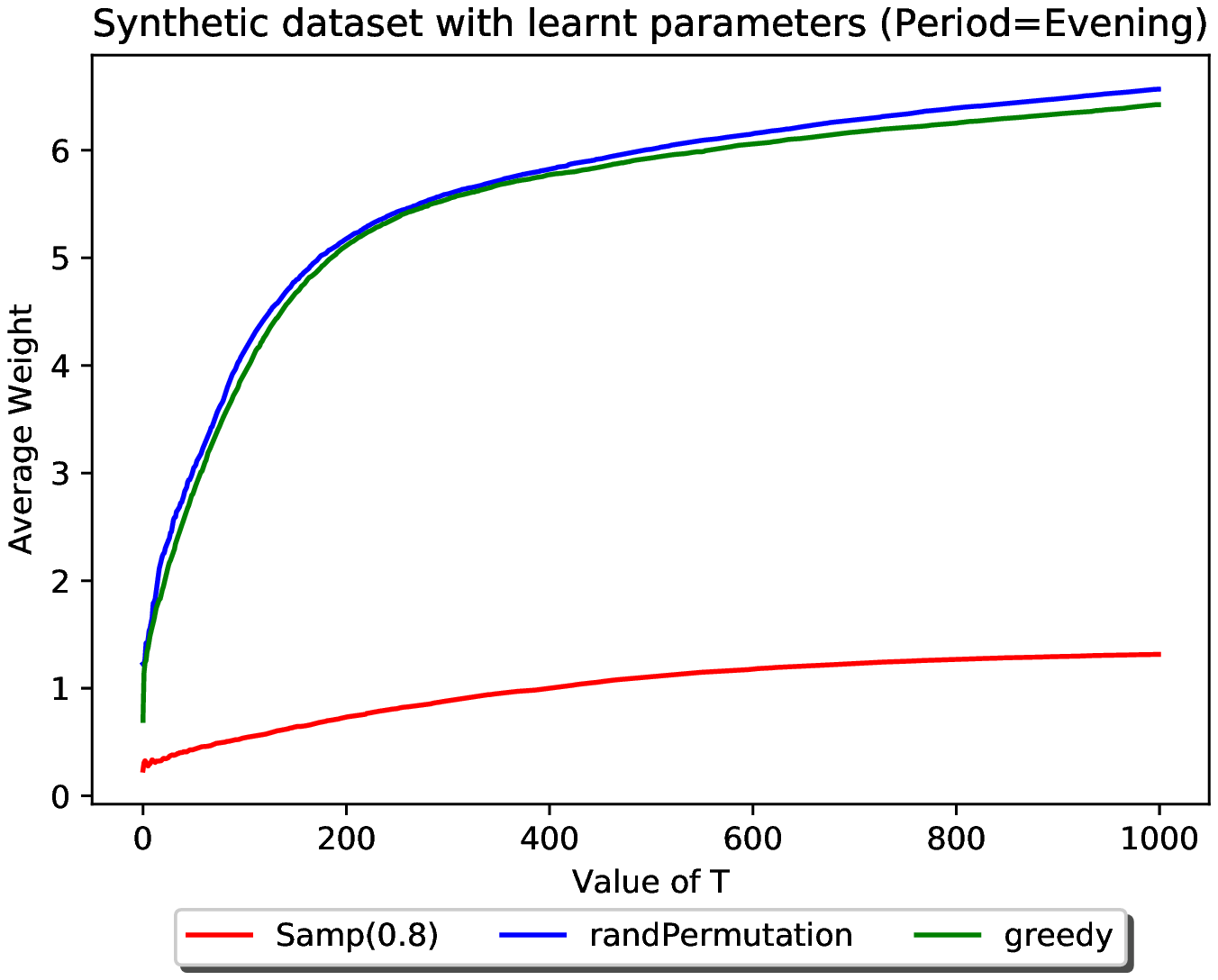}
	\caption{Obj.~\eqref{eqn:OBJ} on simulation with real-data parameters (3:00PM-7:00PM)}
	\label{fig:exp7Evening}
	\endminipage
	 \vspace{5mm}
\end{figure*}

\begin{figure*}[!h]
	\minipage{0.48\textwidth}
	\centering
	 \includegraphics[width=\linewidth]{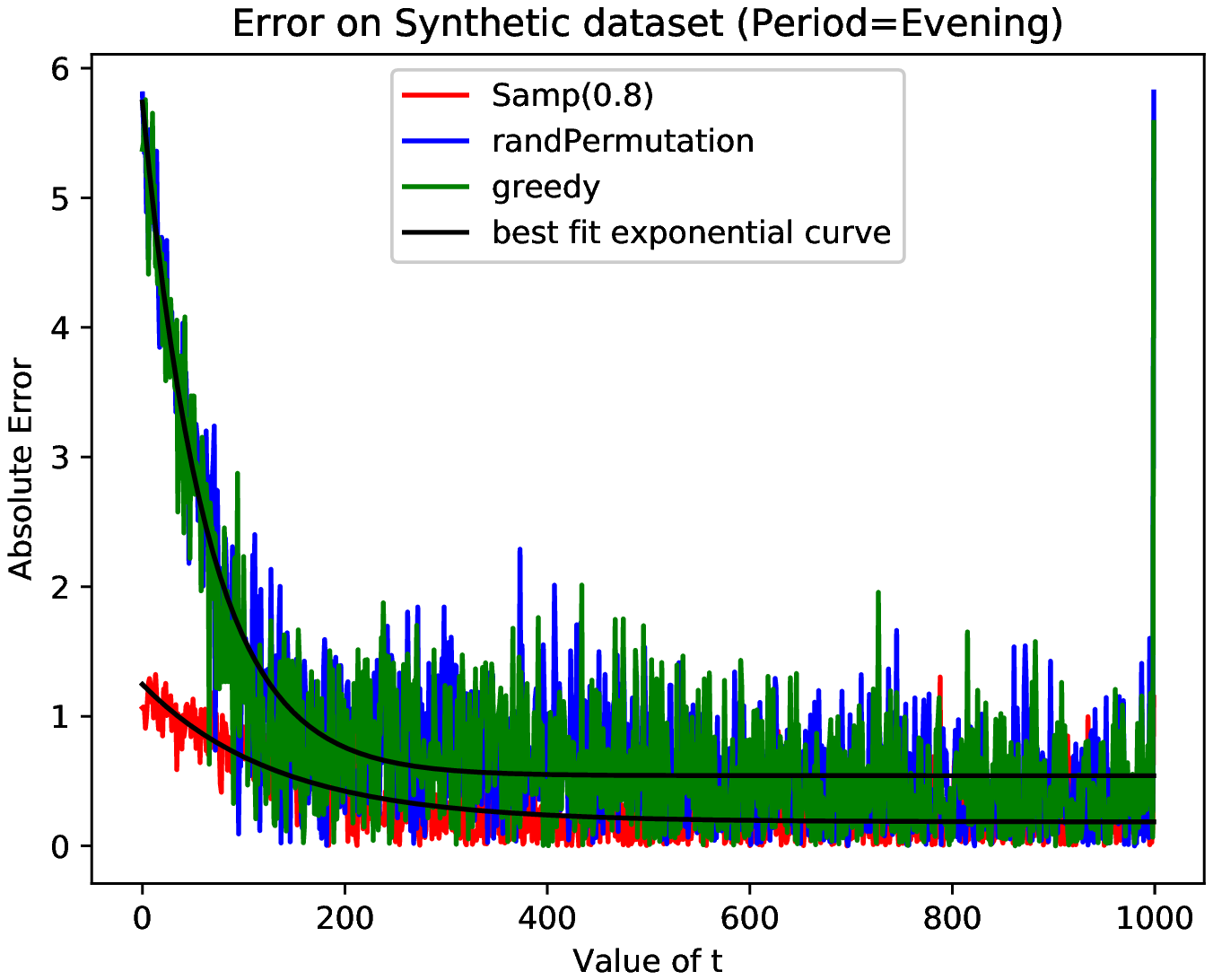}
	 \caption{Error in Obj.~\eqref{eqn:OBJ-t} on simulation with real-data parameters (3:00PM-7:00PM)}
	\label{fig:exp6EveningError}
	\endminipage\hfill
	 \vspace{5mm}
	\minipage{0.48\textwidth}
	\centering
	\includegraphics[width=\linewidth]{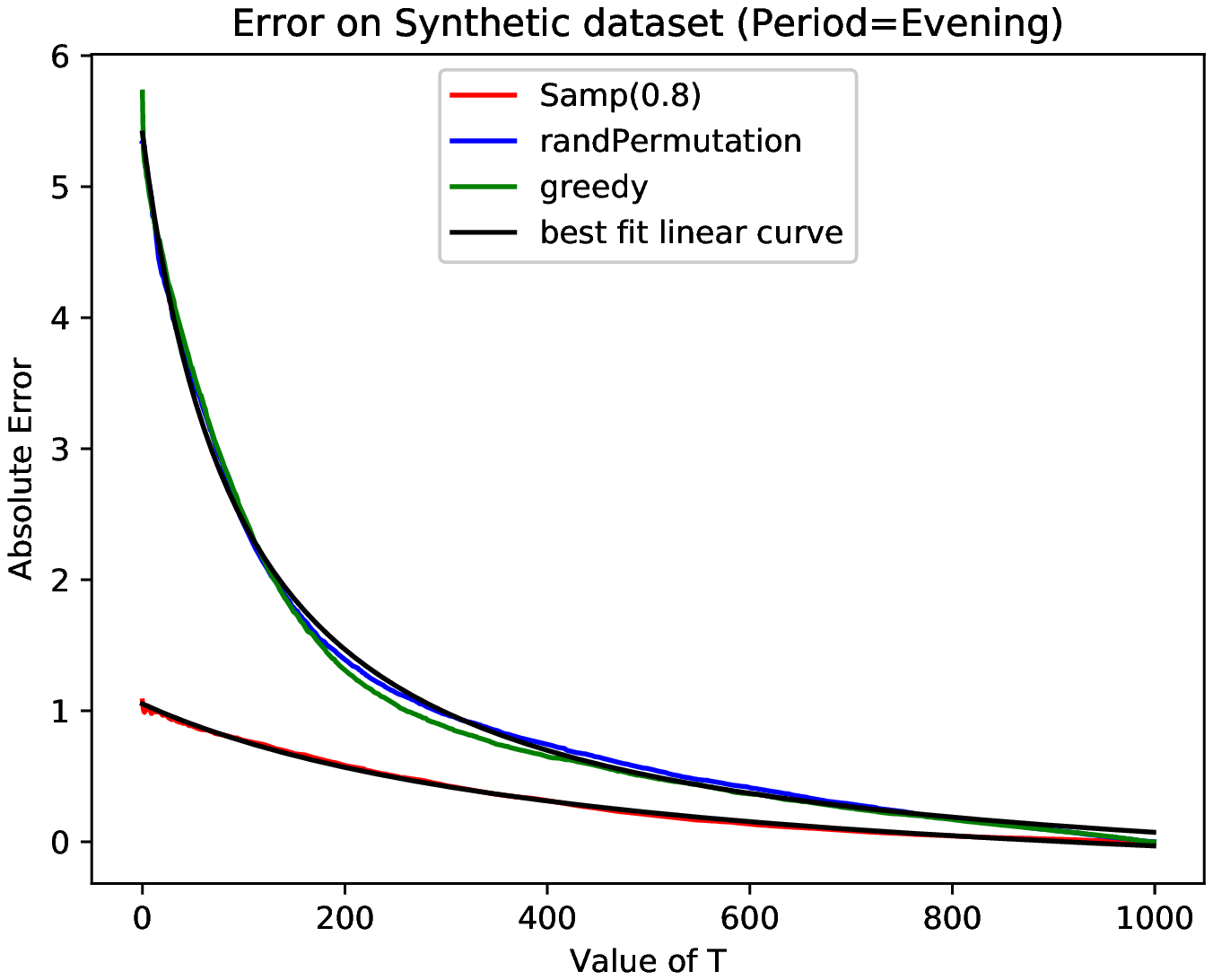}
	\caption{Error in Obj.~\eqref{eqn:OBJ} on simulation with real-data parameters (3:00PM-7:00PM)}
	\label{fig:exp7EveningError}
	\endminipage
	 \vspace{5mm}
\end{figure*}

\begin{figure*}[!h]
	\minipage{0.48\textwidth}
	\centering
	 \includegraphics[width=\linewidth]{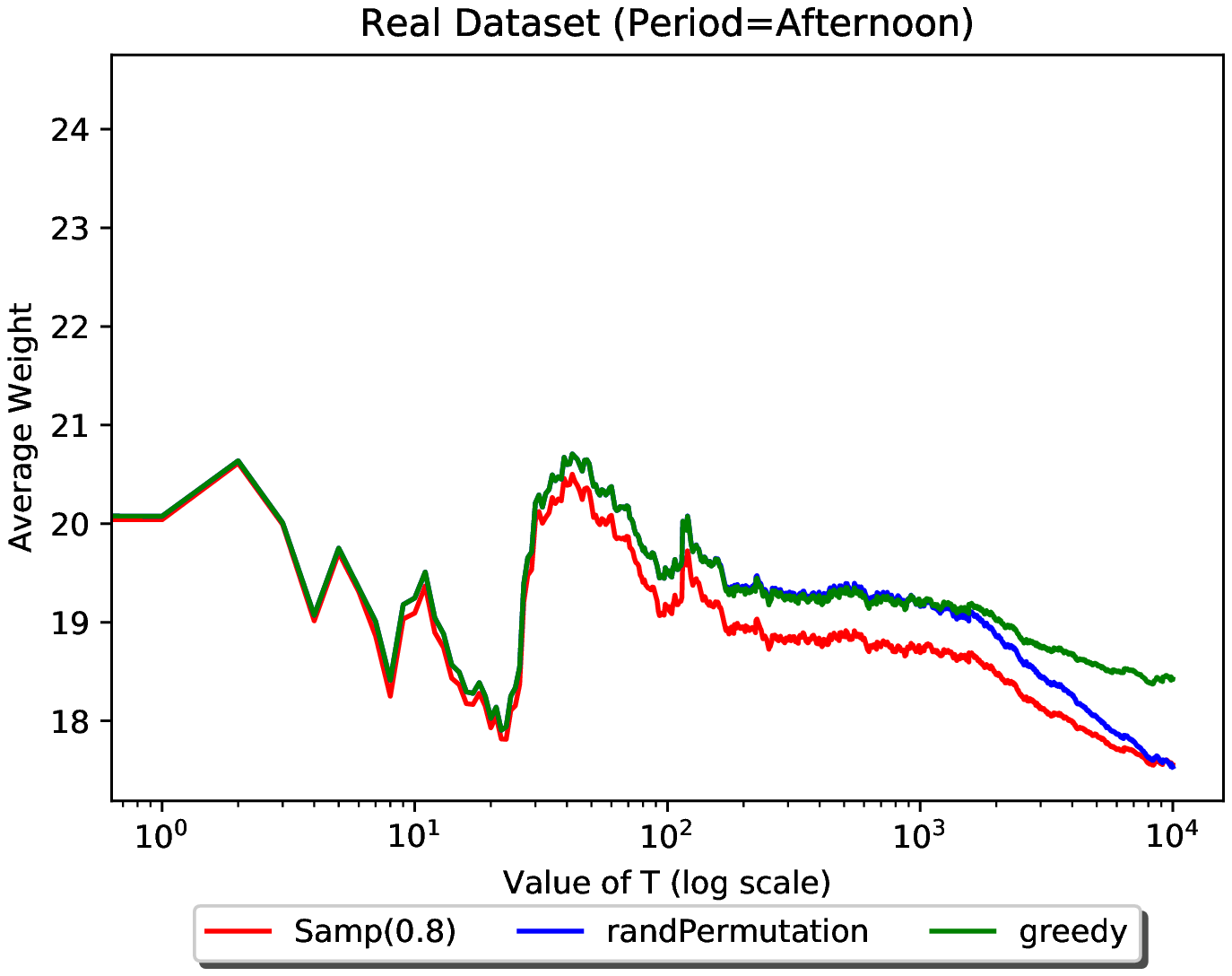}
	 \caption{Obj.~\eqref{eqn:OBJ-t} on real-dataset (11:00AM-3:00AM). $x$-axis in log-scale}
	\label{fig:exp3Morning}
	\endminipage\hfill
	 \vspace{5mm}
	\minipage{0.48\textwidth}
	\centering
	\includegraphics[width=\linewidth]{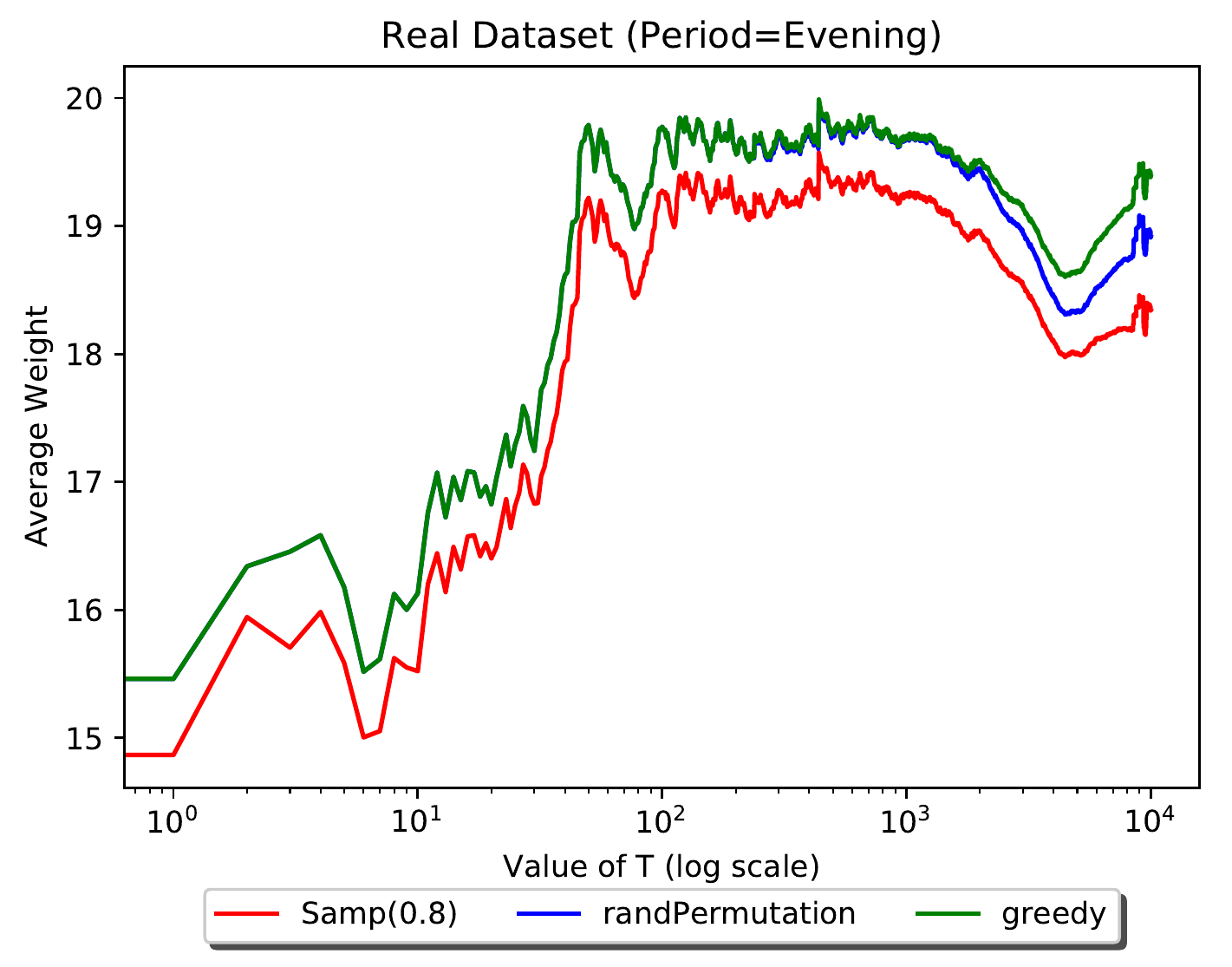}
	\caption{Obj.~\eqref{eqn:OBJ} on real-dataset (3:00PM-7:00PM). $x$-axis in log-scale}
	\label{fig:exp3Evening}
	\endminipage
	 \vspace{5mm}
\end{figure*}

\begin{figure*}[!h]
	\minipage{0.48\textwidth}
	\centering
	 \includegraphics[width=\linewidth]{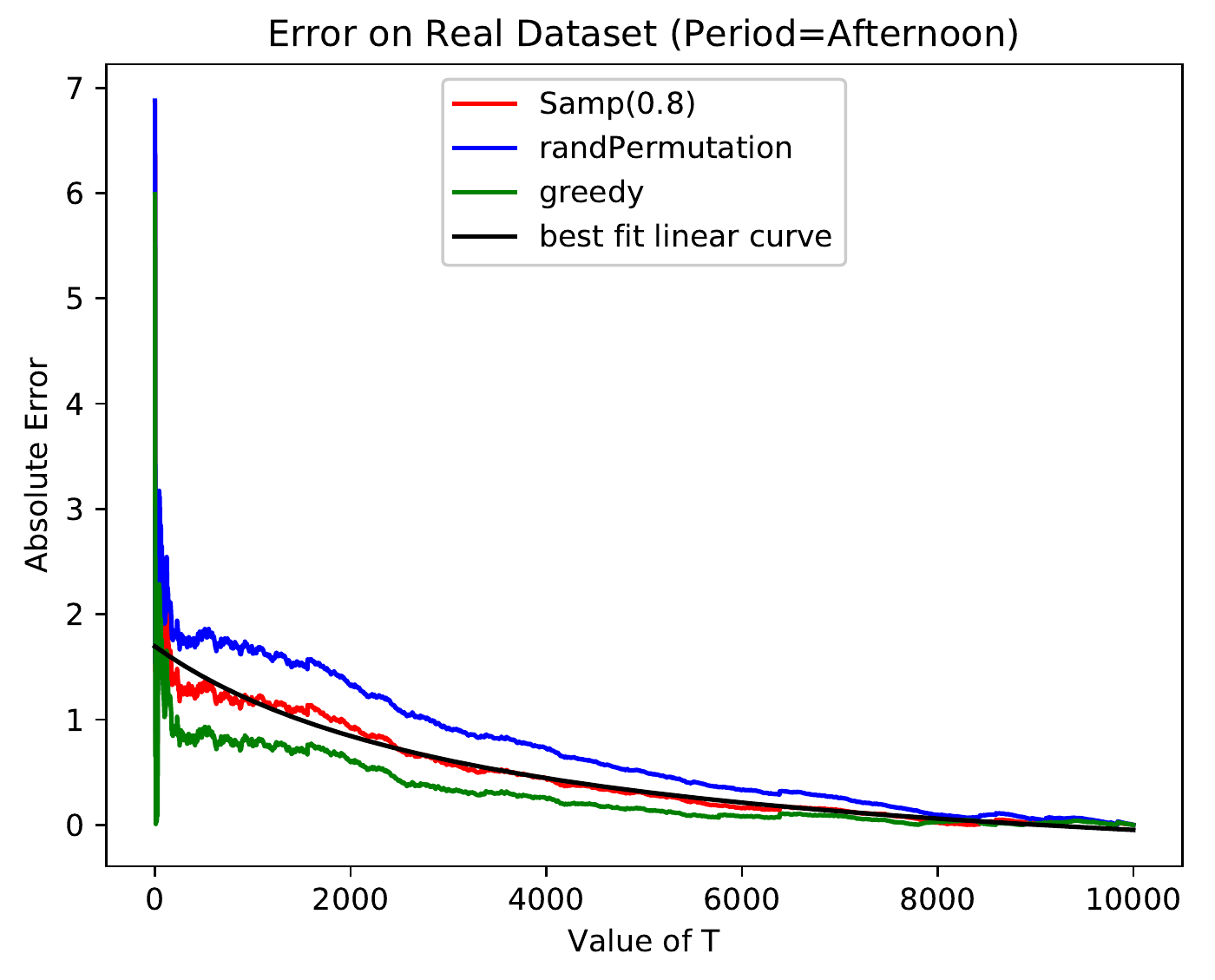}
	 \caption{Error in Obj.~\eqref{eqn:OBJ-t} on real-dataset (11:00AM-3:00PM). $x$-axis in log-scale}
	\label{fig:exp3MorningError}
	\endminipage\hfill
	 \vspace{5mm}
	\minipage{0.48\textwidth}
	\centering
	\includegraphics[width=\linewidth]{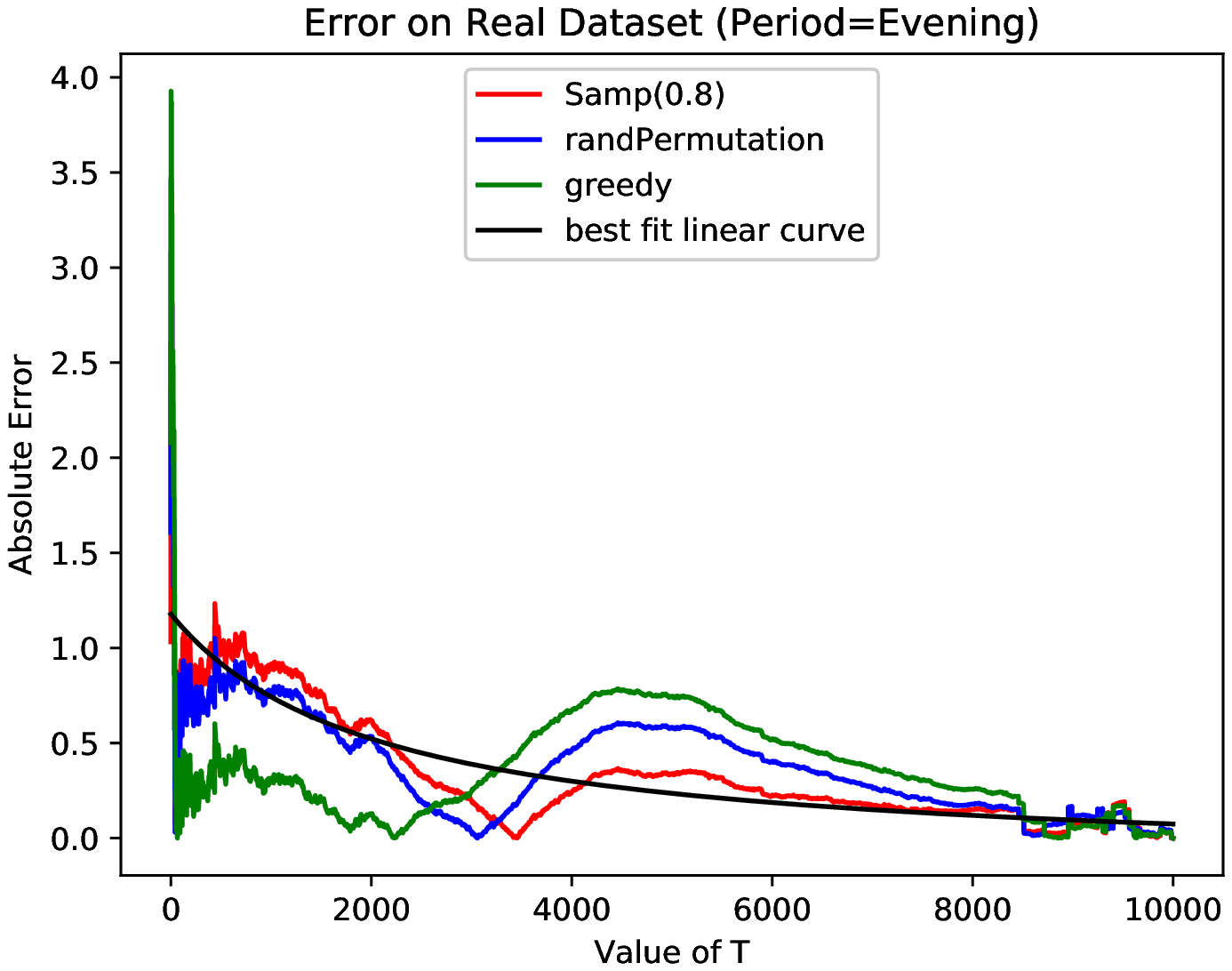}
	\caption{Error in Obj.~\eqref{eqn:OBJ} on real-dataset (3:00PM-7:00PM). $x$-axis in log-scale}
	\label{fig:exp3EveningError}
	\endminipage
	 \vspace{5mm}
\end{figure*}

\begin{figure*}[!h]
	\minipage{0.48\textwidth}
	\centering
	 \includegraphics[width=\linewidth]{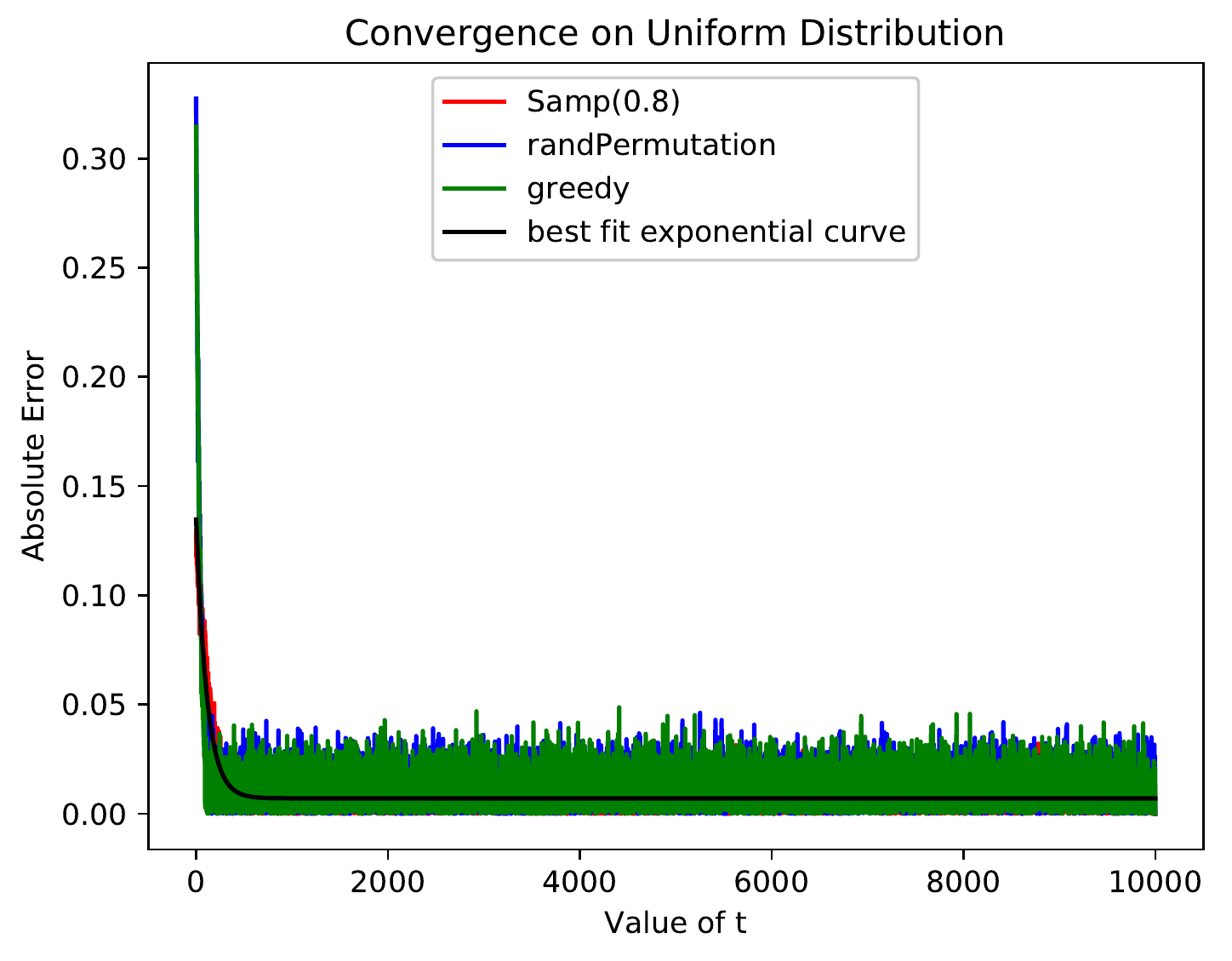}
	 \caption{Error in Obj.~\eqref{eqn:OBJ-t} on simulated data}
	\label{fig:exp4Error}
	\endminipage\hfill
	 \vspace{5mm}
	\minipage{0.48\textwidth}
	\centering
	\includegraphics[width=\linewidth]{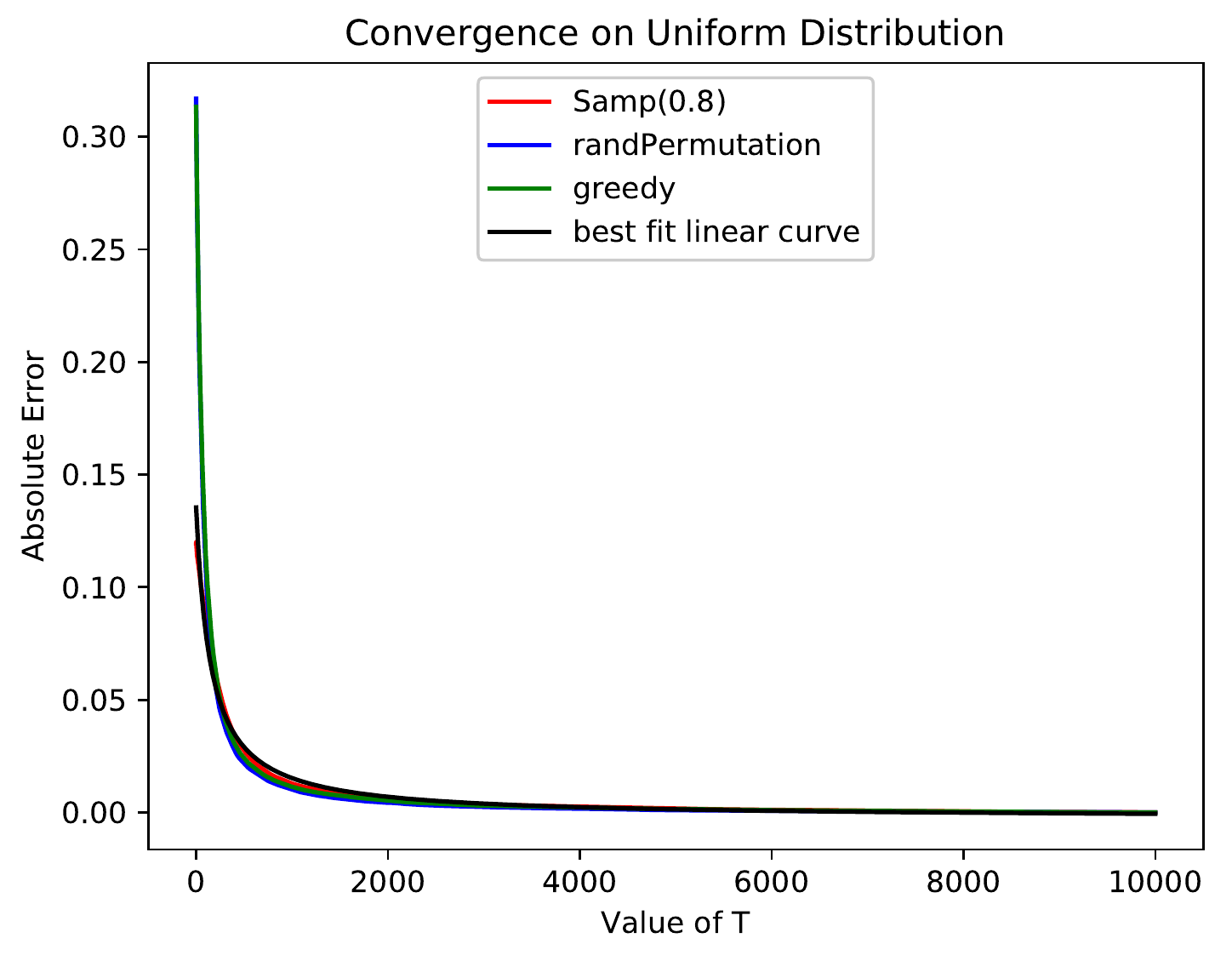}
	\caption{Error in Obj.~\eqref{eqn:OBJ} on simulated data}
	\label{fig:exp4Error}
	\endminipage
	 \vspace{5mm}
\end{figure*}

\section{Details on Deep-Q-Learning}

We used deep Q-learning with experience replay, trained in a custom OpenAI Gym
environment. Memory size was 2000 events, with a batch size of 64 for experience
replay. The discount factor was 0.9. The initial exploration rate was 1.0,
decaying with a decay factor of .995 to a minimum of 0.01. The learning rate for
the Q-learning update was 0.001. Training consisted of 5 episodes, each randomly
initialized and continuing for 2500 steps. The network itself had 2 dense hidden
layers with 24 units and ReLU activation, and a final layer with 6 outputs and a
linear activation; the loss was quadratic. As with the $\nadap(\alpha)$ policy, any
attempt to make a match that would be illegal was simply ignored and given
reward 0.

\section{Details on Value Iteration}
We used standard value iteration algorithm to solve for an exact optimal policy on small grids under uniform IID requests. To illustrate the optimal policy, we present the distribution of different measures of activity as heatmaps on the grid. We initialize the state with 1 car at left-top location and zero car at remaining locations. We then simulate an episode of 1000 periods under the optimal policy by drawing uniform IID requests. The measure ``time covered'' is the percentage of periods in which the location is occupied by at least one car. The measure ``drop rate'' is the percentage of periods in which the location is either a starting or ending location. The measure ``start location'' is the percentage of periods in which the location is a starting location. The reason for considering small grids is  because of the huge number of states. In particular, 2 $\times$ 2 grid with capacity 2, capacity 3 and 3 $\times$ 3 grid with capacity 1 have 324, 1024 and 41472 states respectively.

\begin{figure}[!ht]%
	 \subfloat{{\includegraphics[scale=0.35]{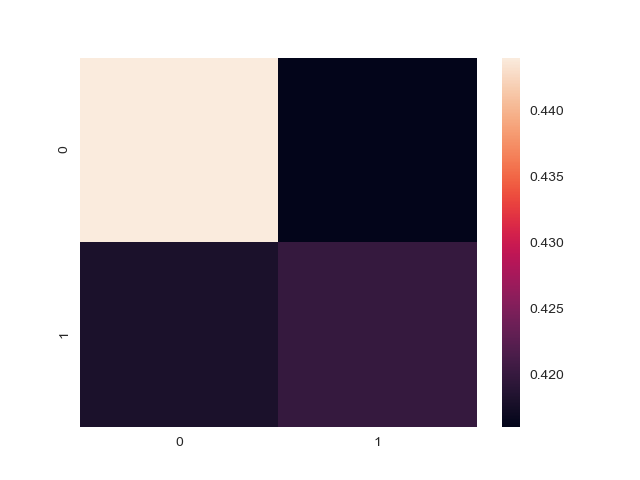} }}%
     \subfloat{ {\includegraphics[scale=0.35]{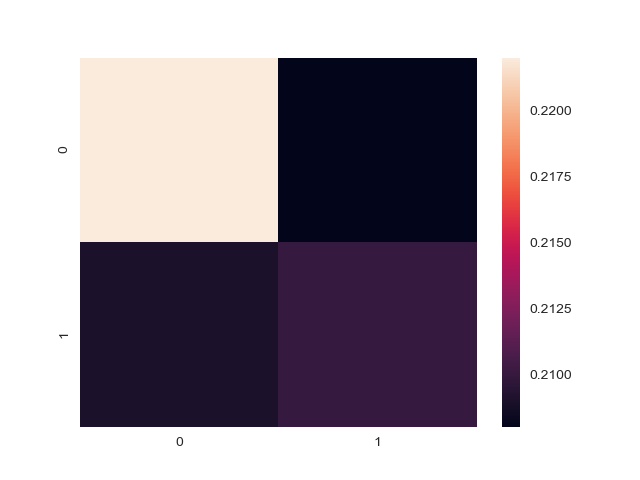} }}%
     \subfloat{ {\includegraphics[scale=0.35]{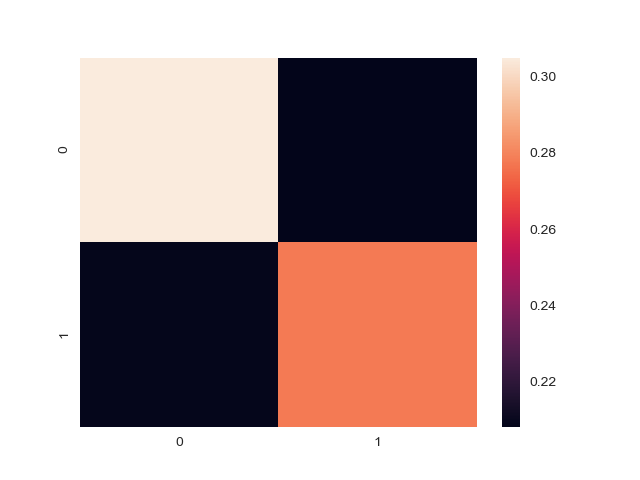} }}%
    \caption{2 $\times$ 2 grids with capacity 3; black to white represent low to high values. Plots 1, 2, 3 represent the distribution of requests with either starting or ending at a cell, requests starting at a cell and percentage of time spent by cars in a cell respectively by the optimal policy.}%
    \label{fig:heatmaps2}%
    \vspace{5mm}
\end{figure}

\begin{figure}[!ht]%
	 \subfloat{{\includegraphics[scale=0.35]{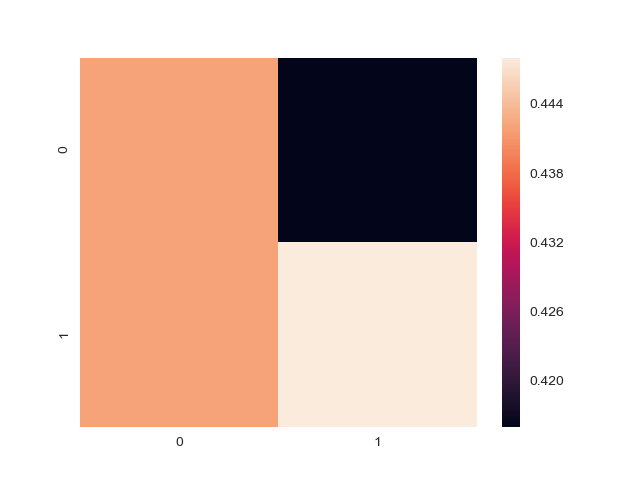} }}%
     \subfloat{ {\includegraphics[scale=0.35]{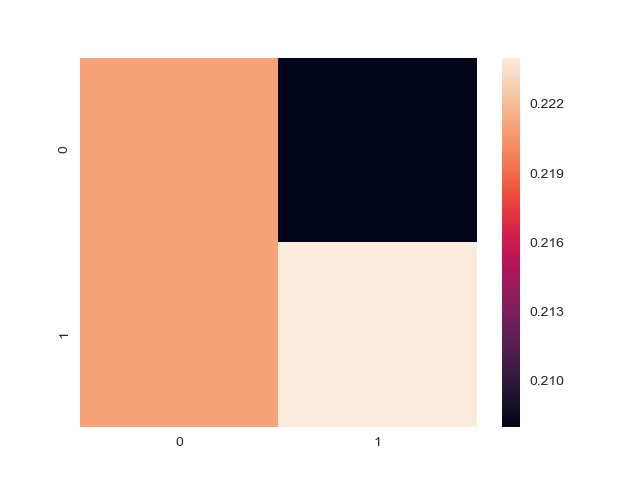} }}%
     \subfloat{ {\includegraphics[scale=0.35]{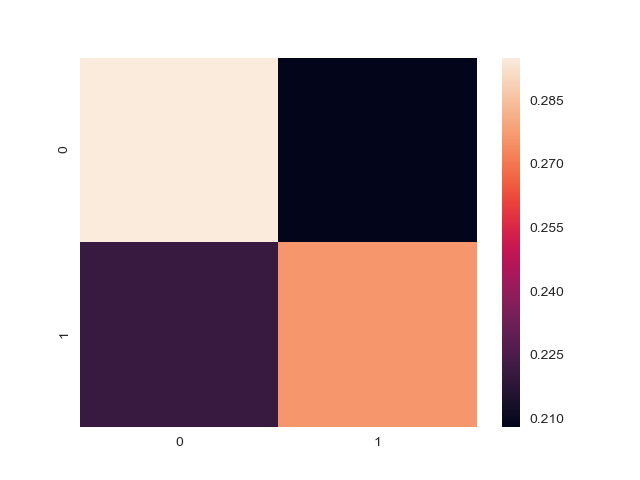} }}%
    \caption{2 $\times$ 2 grids with capacity 2; black to white represent low to high values. Plots 1, 2, 3 represent the distribution of requests with either starting or ending at a cell, requests starting at a cell and percentage of time spent by cars in a cell respectively by the optimal policy.}%
    \label{fig:heatmaps2}%
     \vspace{5mm}
\end{figure}

\bibliographystyle{plainnat}
{ \bibliography{refs}}
	
	\end{document}